\documentclass[journal]{IEEEtran}

\usepackage{multicol}
\usepackage{multirow}
\usepackage{longtable}
\usepackage{subcaption}

\usepackage{balance}
\usepackage{cite}

\usepackage{algorithm}
\usepackage{algpseudocode}

\ifCLASSINFOpdf
   \usepackage[pdftex]{graphicx}
   \DeclareGraphicsExtensions{.pdf,.jpeg,.png}
\else
 
\fi

\usepackage{mathrsfs,amsmath}
\usepackage{amssymb}

\usepackage{amsthm,hyperref,enumerate}
\newtheorem{theorem}{Theorem}

\theoremstyle{definition}
\newtheorem{defn}{Definition}

\theoremstyle{remark}
\newtheorem{remark}{Remark}
\newtheorem{lemma}{Lemma}

\theoremstyle{plain}
 \newtheorem{assumption}{Assumption}

\newcommand{\tabincell}[2]{\begin{tabular}{@{}#1@{}}#2\end{tabular}}

\usepackage{ragged2e}

\usepackage[amssymb]{SIunits}
\hypersetup{
	colorlinks=true,
	linkcolor=blue,
	filecolor=blue,      
	urlcolor=blue,
}

\begin{document}

\title{Integrated Relative-Measurement-Based Network Localization and Formation Maneuver Control (Extended Version)}

\author{Xu Fang,  Lihua Xie, Xiaolei Li
\thanks{This work was supported 
by Nanyang Technological University under the Wallenberg-NTU Presidential Postdoctoral Fellowship and Projects of Major International (Regional) Joint Research Program under NSFC Grant no. 61720106011 and 62103352.}
\thanks{Xiaolei Li is with the school of Electrical Engineering, Yanshan University, Qinhuangdao, China, 066004. This work has been done when this author was with the School of Electrical and Electronic Engineering, Nanyang Technological University, Singapore, 639798.  (E-mail:xiaolei@ysu.edu.cn). }
\thanks{
Xu Fang and Lihua Xie are with the School of Electrical and Electronic Engineering, Nanyang Technological University, Singapore, 639798. (E-mail: fa0001xu@e.ntu.edu.sg; elhxie@ntu.edu.sg).} 
}

\maketitle


\begin{abstract}
This paper studies the problem of integrated distributed network localization and formation maneuver control. 
We develop an integrated
relative-measurement-based 
scheme, which only uses relative positions, distances, bearings, angles, ratio-of-distances, or their combination to achieve distributed network localization and formation maneuver control in $\mathbb{R}^d (d \ge 2)$. By exploring the localizability and invariance of the  target formation, the scale, rotation, and translation of the formation can be controlled simultaneously by only tuning the leaders' positions, i.e., the followers do not need to know parameters of the scale, rotation, and translation of the target formation.
The proposed method can globally drive the formation errors to zero in finite time over multi-layer $d\!+\!1$-rooted graphs. 
A simulation example is given to illustrate the theoretical results.
\end{abstract}


\IEEEpeerreviewmaketitle



\section{Introduction}



\IEEEPARstart{N}{etworked}
multi-agent systems have attracted recurring research interests from the control community due to their extensive  military and civilian applications such as surveillance and cooperative search \cite{ sun2018cooperative,mehdifar20222, an2020, Generalized2021}. There are two fundamental problems in networked multi-agent systems, namely, network localization and formation maneuver control.



On the one hand, network localization aims to determine the unknown positions of some agents (called free nodes) 
by using the known positions of other agents (called anchor nodes) and inter-agent relative measurements.  The
existing distributed network localization methods are classified into five categories: angle-based \cite{jing2019angle1},  bearing-based \cite{zhao2016localizability,li2019globally},   ratio-of-distance-based \cite{fang2020}, distance-based \cite{diao2014barycentric}, and relative-position-based \cite{fang2023distributed}. The existing distributed localization methods \cite{jing2019angle1,  zhao2016localizability, diao2014barycentric, li2019globally,fang2020,fang2023distributed} mainly focus on static sensor networks. 
On the other hand, formation maneuver control aims to change the scale, rotation, and translation of a multi-agent system simultaneously.
Most existing distributed formation maneuver control methods 
\cite{ zhao2018affine,fang20211distributed,   xu2020affine,han2017tc, lin2015necessary}
require the information of inter-agent relative positions. But in many applications, agents can only obtain  non-relative-position measurements such as relative distances, bearings, angles, ratio-of-distances or their combination.
To tackle this problem, scientists try to explore non-relative-position-based integrated distributed localization and formation maneuver control schemes. There are two main challenges:  (\romannumeral1) 
how to handle the situation when agents are not localizable? Note that there is no guarantee that agents are localizable during the transient before they reach their target formation.  For example, to guarantee the localizability of a 2-D angle-based or bearing-based multi-agent system, each agent and its neighbors are required to be non-collinear at all times, which is difficult to be guaranteed during the transient before the agents reach their
target formation; (\romannumeral2) how to design an integrated estimation and control scheme to achieve a desired formation?

\begin{table*}[t]
    \begin{center}
    \caption{Comparison with existing integrated distributed network localization and formation maneuver control.}
        \begin{tabular}{c|c c c}
        \hline
      Methods   & Measurements  & Constraints & Advantages of our method
      \\ \hline
          Distance-based \cite{cao2011formation,jiang2016simultaneous} &   Distance
         & \tabincell{c}{ $\mathcal{A}1$, 2-D space, constant velocity of the leader, \\  circular motion of each agent}
         & $\mathcal{A}2$, 3-D space, no constraint on the agents   \\ 
        \hline
         Distance-based \cite{nguyen2019persistently, cao2019relative,han2018integrated} & \tabincell{c}{Distance, odometry, \\ derivative of distance}
         &  \tabincell{c}{$\mathcal{A}1$, persistently excited motion of the agents, \\ zero initial localization errors }   & \tabincell{c}{$\mathcal{A}2$, only need distance measurements, \\ no constraint on the agents} \\ 
         \hline
          Bearing-based \cite{yang2020distributed} & Bearing
         &  \tabincell{c}{$\mathcal{A}1$, 2-D space, \\ constrained initial positions of the leaders}
        &  \tabincell{c}{$\mathcal{A}2$, 3-D space, \\ no constraint on the agents}
        \\
         \hline
         Angle-based \cite{chen2022simultaneous}  & Angle
        &  $\mathcal{A}1$, 2-D space
        & $\mathcal{A}2$, 3-D space
         \\
        \hline        Ratio-of-distance-based  & Ratio-of-distance & $-$ & Solved in this article \\
        \hline 
       Mixed-measurement-based  & Mixed measurements & $-$ & Solved in this article \\
        \hline
        \end{tabular}
\label{tab:accuracy-orb}
\\~\\

\begin{flushleft}
$\mathcal{A}1$: All agents need to know the maneuver parameters (scale, rotation, and translation of the formation) or the followers need to estimate the maneuver parameters only known to the leaders; 

$\mathcal{A}2$: Only the leaders know the maneuver parameters, and the followers do not need to know the maneuver parameters;

$-$: To the best of our knowledge, 
there exists no result for 3-D angle-based, ratio-of-distance-based, or mixed-measurement-based integrated distributed  localization and  formation maneuver control. The term "mixed-measurement" means that the followers are allowed to have different types of relative measurements, e.g., some agents can measure only bearings, while others can measure only ratio-of-distances, angles, or distances.

\end{flushleft}


\end{center}
\end{table*}

The non-relative-position-based  integrated distributed localization and formation maneuver control are deeply-investigated in \cite{cao2011formation, jiang2016simultaneous,nguyen2019persistently,han2018integrated, cao2019relative,yang2020distributed,chen2022simultaneous}. The problem is that their methods can
only be applied in 2-D space or need to impose constraints on the motion of the agents (e.g., persistently excited motion of the agents) shown in Table \ref{tab:accuracy-orb}. If it is in 3-D space or there is no constraint on the motion of the agents, the relative positions among the agents cannot be estimated by their proposed distance-based \cite{cao2011formation, jiang2016simultaneous,nguyen2019persistently,han2018integrated, cao2019relative}, bearing-based \cite{yang2020distributed}, or angle-based \cite{chen2022simultaneous} estimators, and thus the formation maneuver control cannot be achieved. The work in \cite{cai2019integrated} reveals how control errors and localization errors affect formation accuracy.
In addition, the existing distributed methods \cite{cao2011formation, jiang2016simultaneous, chen2022simultaneous, nguyen2019persistently, cao2019relative, han2018integrated,  yang2020distributed} assume that the followers 
have identical relative measurement type and the information of time-varying maneuver parameters.

Motivated by the above limitations, we develop a novel relative-measurement-based two-mode switching control scheme, which
can not only overcome the limitations of the existing methods \cite{cao2011formation, jiang2016simultaneous,nguyen2019persistently,han2018integrated, cao2019relative,yang2020distributed,chen2022simultaneous} shown in Table \ref{tab:accuracy-orb}
but also solve the unsolved 3-D angle-based,  ratio-of-distance-based, or mixed-measurement-based integrated distributed localization and formation maneuver control problem. The contributions of this article lie in the following aspects: 
\begin{enumerate}[(1)]
    \item An integrated relative-measurement-based  distributed localization and formation maneuver scheme in $\mathbb{R}^d (d \! \ge \! 2)$ is proposed, where 
    each follower can measure one of the following five types of relative measurements:  relative position, bearing, distance, angle, or ratio-of-distance;
    \item The scale, rotation, and translation of the formation can be controlled simultaneously by only tuning the leaders' positions, while the followers do not need to know the time-varying maneuver parameters of the target formation; 
    \item A mode control scheme (maneuvering mode and maintaining mode)
    is proposed to overcome the challenge that agents may not be localizable at all times and achieve the target formation.

    
\end{enumerate}

The rest of this paper is organized as follows. 
The problem statement is given in Section \ref{problem}.
Section \ref{displace} presents the concept of displacement constraint and design of target formation.  Section \ref{mmm} introduces two operation modes for the followers.
The integrated distributed localization and formation maneuver control is given in Section \ref{single}. A simulation example is shown in Section \ref{simul} to verify the effectiveness of the proposed approach. Concluding remarks are given in Section \ref{concl}.

\section{Problem Statement}\label{problem}

\subsection{Notations}



\begin{figure}[t]
\centering
\includegraphics[width=0.85\linewidth]{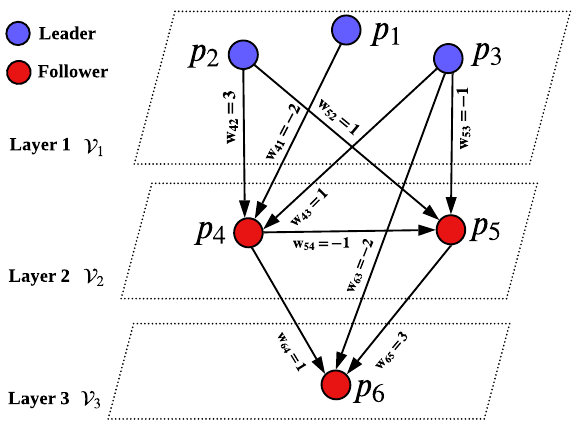}
\caption{A multi-layer $3$-rooted graph in 2-D space. The graph is not a multi-layer $3$-rooted graph if we remove the edge between leader $3$ and follower $6$.} 
\label{2d}
\end{figure}

The position of agent $i$ in $\mathbb{R}^d (d \! \ge \! 2)$ is denoted by $p_i$. Let $p=[p_1^T,\cdots, p_n^T]^T$ be a configuration of $n$ agents. A leader-follower formation of $n$  agents in $\mathbb{R}^d$ is denoted by $(\mathcal{G}, p)$, where $\mathcal{G}=\{ \mathcal{V},\mathcal{E}\}$ is the communication graph consisting of an agent set $\mathcal{V}=\{1, \cdots, n \}$ and an edge set $\mathcal{E} \subseteq \mathcal{V} \times \mathcal{V}$. Let $\mathcal{V}_l =\{1, \cdots, m \}$ and $\mathcal{V}_f\!=\! \{m\!+\!1, \cdots, n\}$ be the set of the leaders and followers, respectively. The edge $(i,j) \in \mathcal{E}$ indicates that agent $i$ can obtain information from agent $j$. The different kinds of inter-agent relative measurements are
\begin{equation}
\begin{array}{ll}
     &  e_{ij} \!=\!  p_j -p_i, \ \ \ \ \ \ g_{ij}\!=\! \frac{p_j -p_i}{d_{ij}},  \\
     & d_{ij}\!=\! \|p_j-p_i \|_2, \ \  \theta_{ijk}\!=\! \arccos (g^T_{ij}g_{ik}),
\end{array}
\end{equation}
where $e_{ij}, g_{ij}, d_{ij}$ are, respectively, the relative position, bearing, and distance between agent $i$ and agent $j$. 
{$\theta_{ijk} \! \in \! [0, \pi]$} is the angle between $g_{ij}$ and $g_{ik}$. $\| \cdot \|_2$ is the $\mathcal{L}_2$ norm.  Denote $\text{SO}(d)$  as the set of rotation matrices in $\mathbb{R}^d$. Let ${\mathbf{0}}_d , {\mathbf{1}}_d \in \mathbb{R}^{d}$ be the vector with all entries equal to zero and one, respectively. Let ${I}_d \in \mathbb{R}^{d \times d}$ be the identity matrix of dimension $d \times d$. Let $\text{Rank}(\cdot)$ be the rank of a matrix. An ordered sequence of agents $v_{i_1}, \cdots, v_{i_j}$ with $(v_{i_k}, v_{i_{k\!+\!1}}) \in \mathcal{E}$ for $k=1, \cdots, j-1$ is called a walk. A path is a walk without repeated agents.
Agent $i \in \mathcal{V}$ is 
called $d\!+\!1$-reachable from a set $V_c \in \mathcal{V}$ if there exists a path from 
$V_c$ to $i$ after removing any $d$ agents except agent $i$, i.e., there 
are $d+1$ disjoint paths from $V_c$ to $i$\cite{lin2015necessary}.

\begin{defn}\label{mul}
A graph $\mathcal{G}$ is called a multi-layer $d\!+\!1$-rooted graph in $\mathbb{R}^d (d \ge 2)$ if
\begin{enumerate}[(i)]
\item The agents in $\mathcal{G}$ are divided into $\kappa \!>\! 1$ subsets $\mathcal{V}_1,  \mathcal{V}_{2},$ $\cdots, \mathcal{V}_{{\kappa}}$, where $\mathcal{V}_i \cap \mathcal{V}_j = \emptyset$ if $i \neq j$. Agent $i$ is called in layer $g$ if  $i \in \mathcal{V}_{g} (1 \le g \le \kappa)$. Subset $\mathcal{V}_1$ includes all leaders, i.e., $\mathcal{V}_1=\mathcal{V}_l$. The union of the subsets $\mathcal{V}_{2}, \cdots, \mathcal{V}_{{\kappa}}$ include all followers, i.e.,  $\bigcup\limits_{g=2}^{\kappa} \mathcal{V}_g = \mathcal{V}_f$;
\item Each follower $i$
is $d\!+\!1$-reachable from the leader set $V_1$, and its neighbor set ${N}_i$ is given by
\begin{equation}\label{neighborset}
\mathcal{N}_i = \{ j \in \bigcup\limits_{s=1}^{g} \mathcal{V}_{s}:  j <i, \ (i,j) \in \mathcal{E}, \ i \in \mathcal{V}_{g} \}.
\end{equation} 
\end{enumerate}
\end{defn}

\begin{remark}
    A simple multi-layer $d\!+\!1$-rooted graph in $\mathbb{R}^2$ is given in Fig. \ref{2d}. The difference of "multi-layer $d \!+\! 1$-rooted graph" and "$d \!+\! 1$-rooted graph" \cite{han2017tc} is that "multi-layer $d \!+\! 1$-rooted graph" decouples the graph into several 
hierarchical layers, where the hierarchical decomposition algorithm is given in (\romannumeral1) and (\romannumeral2) of Definition \ref{mul}. The minimum number of neighbors of each follower in a multi-layer $d\!+\!1$-rooted graph is $d\!+\!1$ in $\mathbb{R}^d$.
\end{remark}

\subsection{Target Formation and Control Objective}

Let $p^*_l \!=\! [{p^*_1}^T, \cdots, {p^*_{m}}^T]^T \! \in \! \mathbb{R}^{md}$ and $p^*_f \!=\! [{p^*_{m\!+\!1}}^T , \cdots, {p^*_{n}}^T]^T \! \in \! \mathbb{R}^{(n-m)d}$ be the target positions of the leaders and followers in $\mathbb{R}^d$, respectively. Denote $p^*(t)=[{p^*_l}^T\!(t),{p^*_f}^T\!(t)]^T$ as a configuration of the time-varying target formation $(\mathcal{G}, p^*(t))$.
The time-varying target formation $(\mathcal{G}, p^*(t))$ is designed based on a constant nominal formation $(\mathcal{G}, r)$, i.e.,
\begin{equation}\label{ti}
    p^*(t)= \beta(t)[I_n \otimes Q(t)]r+ {\mathbf{1}}_n \otimes \delta(t),
\end{equation}
where $\beta(t) \in \mathbb{R}$, $Q(t) \in SO(d)$, and $\delta(t) \in \mathbb{R}^d$ are, respectively, time-varying scale, rotation, and translation parameter. $r=[r^T_1, \cdots, r^T_n] \! \in \! \mathbb{R}^{nd}$ is a constant nominal configuration to be designed later. The available information to the agents is: 
\begin{enumerate}[(i)]
    \item The time-invariant nominal configuration $r$;
    \item  Only the leaders know their own positions and the maneuver parameters $\beta(t), Q(t), \delta(t)$ in \eqref{ti};
    \item  Each follower has at least $d\!+\!1$ neighbors in $\mathbb{R}^d$, and can measure one 
    of the following five types of relative measurements: ratio-of-distance, angle. distance, bearing, or relative position.
\end{enumerate}

In this article, the followers are allowed to have different types of relative measurements, e.g., some agents can measure only bearings, while others can measure only ratio-of-distances, angles, or distances.
We consider that each agent is governed by a single-integrator dynamics.
\begin{equation}\label{singlem}
    \dot p_i = u_i, \ i=1, \cdots, n,
\end{equation}
where $u_i \in \mathbb{R}^d$ represents the control input of agent $i$. 
The control objectives of the leaders and followers are given by
\begin{align}\label{con1l}
     &  \ \ \ \ \ \  \lim\limits_{t \rightarrow \infty} (p_i(t)-p^*_i(t)) = \mathbf{0}, \ \ i \in \mathcal{V}_l,     \\
     & \begin{array}{ll}\label{con1}
    &  \left\{ \! \begin{array}{lll} 
       \lim\limits_{t \rightarrow \infty} (\hat p_i(t) - p_i(t)) = \mathbf{0}, \\
        \lim\limits_{t \rightarrow \infty} (p_i(t) - p^*_i(t)) = \mathbf{0}, 
    \end{array}\right.  i \in \mathcal{V}_f,
\end{array} 
\end{align}
where $\hat{p}_i(t)$ is the position estimate of agent $i$, and the target position
$p^*_i(t)$ of agent $i$ is the $i$-th element of $p^*(t)$ designed in \eqref{ti}, i.e.,
\begin{equation}\label{elem}
   p_i^*(t)= \beta(t) Q(t)r_i+ \delta(t).
\end{equation}

We aim to achieve the control objectives \eqref{con1l} and \eqref{con1} over multi-layer $d\!+\!1$-rooted graphs given in Definition \ref{mul}, where each follower obtains information by communicating with its neighbors or
self-sensing.

\section{Displacement Constraint and Design of Target Formation}\label{displace}

\subsection{Displacement Constraint}\label{displacon}

The relative positions between 
agent $i$ and its any $d\!+\!1 (d \! \ge \! 2)$ neighbors ${j_0}, \cdots, {j_{d}}$ are denoted by $e_{ij_0},  \cdots, e_{ij_d}$, respectively. 
Let $E_i :=
[e_{ij_0} |e_{ij_1} | \cdots |e_{ij_d} ] \in R^{d \times (d+1)}$ and $h_i\!=\! (h_{ij_0}, \cdots,  h_{ij_d})^T \in \mathbb{R}^{d\!+\!1}$. Since the number of elements in the vector $h_i$ is $d\!+\!1$ and \text{Rank}$(E_i) < d\!+\!1$,
there must exist a nonzero vector $h_i \neq \mathbf{0}$ such that $E_ih_i= \mathbf{0}$, i.e.,
\begin{equation}\label{root}
\sum\limits^{d}_{k=0}  h_{ij_k}e_{ij_k} \!=\! \mathbf{0},
\end{equation}
where $h_{ij_0},  \cdots,  h_{ij_d}$ are the displacement parameters and
\begin{equation}\label{fac1}
\sum\limits^{d}_{k=0} h^2_{ij_k} \neq 0.
\end{equation}

\begin{defn}
Equation \eqref{root} is defined as a displacement constraint for agent $i$ and its $d\!+\!1$ neighbors in $\mathbb{R}^d$.
\end{defn}

The displacement parameters $h_{ij_0},  \cdots,  h_{ij_d}$ in \eqref{root} describe the time-varying geometric relationship among the real-time positions of agent $i$ and its $d\!+\!1$ neighbors ${j_0}, \cdots, {j_{d}}$, which can be calculated by relative positions,  bearings, distances, angles, ratio-of-distances, or their combination among agent $i$ and its $d\!+\!1$ neighbors ${j_0}, \cdots, {j_{d}}$ (the details can be found in Appendix and \cite{fang20203}). 
From \eqref{root}, it has
\begin{equation}\label{zero}
    h_{ii}p_i \!=\! \sum\limits^{d}_{k=0} h_{ij_k}p_{j_k},
\end{equation}
where $h_{ii}= \sum\limits^{d}_{k=0} h_{ij_k}$. If $h_{ii} \neq 0$, agent $i$ can be localized by its $d\!+\!1$ neighbors ${j_0}, \cdots, {j_{d}}$, i.e.,
\begin{equation}\label{linear}
\begin{array}{ll}
    p_i \!=\! \sum\limits^{d}_{k=0} \frac{h_{ij_k}}{h_{ii}}p_{j_k}.
\end{array}
\end{equation}

\begin{defn}\label{locade}
Equation \eqref{root} with $h_{ii} \neq 0$  is formally defined as a localizable displacement constraint for agent $i$ and its $d\!+\!1$ neighbors in $\mathbb{R}^d$.
\end{defn}

\begin{defn}
An ambient space is the space surrounding a mathematical object along with the object itself.
A hyperplane in $\mathbb{R}^{d}$ is a subspace
whose dimension is one less than that of its ambient space, which can be described by
\begin{equation}\label{hy}
  z^Tb = c, 
\end{equation}
where $c \in \mathbb{R}$ is a constant and $b \in \mathbb{R}^d$ is a non-zero vector. $z \in \mathbb{R}^d$ is any point on the hyperplane. For example, the hyperplane is a line in 2-D space, or a plane in 3-D space. 
\end{defn}

\begin{lemma}\label{ls2}
Each agent $i$ and its $d\!+\!1$ neighbors ${j_0}, {j_1}, $ $\cdots, {j_{d}}$ can form a localizable displacement constraint if $p_{j_0},p_{j_1},  \cdots, p_{j_{d}}$ are not on a hyperplane in $\mathbb{R}^{d}$.
\end{lemma}

\begin{proof}
We can prove Lemma \ref{ls2} by contradiction. If the displacement constraint \eqref{root} is not localizable, we have $h_{ii} = \sum\limits^{d}_{k=0}  h_{ij_k}=0$, i.e.,
\begin{equation}\label{paraw}
h_{ij_d} \!=\! -
\sum\limits^{d-1}_{k=0} h_{ij_k}.
\end{equation}

Combining \eqref{zero} and \eqref{paraw}, it has
\begin{equation}\label{paraw3}
\sum\limits^{d-1}_{k=0}h_{ij_k}e_{j_dj_k}
= \mathbf{0}.  
\end{equation}

It is concluded from $\sum\limits^{d}_{k=0} h^2_{ij_k} \neq 0$ and \eqref{paraw}  that
\begin{equation}\label{paraw2}
\sum\limits^{d-1}_{k=0}  h^2_{ij_k} \neq 0.
\end{equation}

Then, we can know from \eqref{paraw3} and \eqref{paraw2} that $e_{j_dj_0}, e_{j_dj_1}, \cdots, e_{j_dj_{d\!-\!1}}$ are linearly dependent. Let $F_j \!=\! [ e_{j_dj_0}, e_{j_dj_1}, \cdots, e_{j_dj_{d\!-\!1}}] \! \in \! \mathbb{R}^{d \times d}$. Since $e_{j_dj_0},e_{j_dj_1}, $ $\cdots, e_{j_dj_{d\!-\!1}}$ are linearly dependent, the matrix $F_j$ is not full rank.  Since $\text{Rank}(F_j^T)= \text{Rank}(F_j)$, the matrix $F_j^T$ is also not full rank. From the matrix theory, there must exist a non-zero vector $b \in \mathbb{R}^d$ such that $F_j^T b = \mathbf{0}$, i.e.,
\begin{equation}
\begin{array}{ll}
\left[ \!
	\begin{array}{ll}
e_{j_dj_0}, \
e_{j_dj_1}, \
	\cdots, \
e_{j_dj_{d\!-\!1}} \! \!
	\end{array}
	\right]^T b = \mathbf{0}.
\end{array}
\end{equation}

Then, it has
\begin{equation}\label{li}
\begin{array}{ll}
\left[ \!
	\begin{array}{ll}
	p_{j_0}, \
	p_{j_1}, \
	\cdots, \
	p_{j_{d\!-\!1}} \! \!
	\end{array}
	\right]^T b = p_{j_d}^Tb \otimes \mathbf{1}_d.
\end{array}
\end{equation}

From \eqref{li}, we can know that $p_{j_0}, p_{j_1},  \cdots, p_{j_{d}}$ are on the following hyperplane:
\begin{equation}
    z^Tb = c,
\end{equation}
where $c =p_{j_d}^Tb$. Hence, if $p_{j_0}, p_{j_1},  \cdots, p_{j_{d}}$ are not on a hyperplane in $\mathbb{R}^{d}$, the displacement constraint \eqref{root} is localizable.
\end{proof}

\subsection{Design of Target Formation}\label{local}

It is shown in \eqref{ti} that
the time-varying target formation $(\mathcal{G}, p^*(t))$ is designed based on a constant nominal formation $(\mathcal{G}, r)$. We will first introduce how to design a constant nominal formation $(\mathcal{G}, r)$. 
Let ${r_l} \!=\! [r_1^T, \cdots, r_{m}^T]^T$
and $ {r_f} \!=\! [r_{m\!+\!1}^T , \cdots, r_{n}^T]^T$ be the nominal positions of the leaders and followers, respectively. Suppose there are $d\!+\!1$  leaders in $\mathbb{R}^d$, i.e., $m \!=\! d\!+\!1$. 
From \eqref{root}, for the nominal positions of each follower $i$ 
and the chosen $d\!+\!1$ neighbors $ {j_0},{j_1}, \cdots, {j_d} \in \mathcal{N}_i$ in $\mathbb{R}^d$, we can construct a displacement constraint 
\begin{equation}\label{2di}
\sum\limits^{d}_{k=0}  w_{ij_k}r_{ij_k} \!=\! \mathbf{0}, \ i \in \mathcal{V}_f,
\end{equation}
where $r_{ij}\!=\!r_j\!-\!r_i$, and $w_{ij_0},  \cdots,  w_{ij_d}$ are the displacement parameters. Then, \eqref{2di} can be rewritten as
\begin{equation}\label{we}
w_{ii}r_i \!=\! \sum\limits^{d}_{k=0} w_{ij_k}r_{j_k}, 
\end{equation}
where $w_{ii}\!=\! \sum\limits_{k = 0}^{d}w_{ik} 
$. There will be $n\!-\!m$ displacement constraints for the follower group, which can be written in a compact form as
\begin{equation}\label{form}
(\Omega_f \otimes I_d)r=\mathbf{0},  
\end{equation}
where $r=[r^T_1, \cdots, r^T_n] \! \in \! \mathbb{R}^{nd}$ is a constant nominal configuration and $\Omega_f \! \in \! \mathbb{R}^{(n\!-\!m)\times (n\!-\!m)}$ is called the follower matrix satisfying
\begin{equation}\label{zer}
\begin{array}{ll}
    &[\Omega_f]_{ij} \!=\!  \left\{ \! \begin{array}{lll} 
    -w_{ij}, & 
     j \in \mathcal{N}_i, \ j \neq i, \\
    \ \  0, &   j \notin \mathcal{N}_i, \ j \neq i, \\
    \sum\limits_{j \in \mathcal{N}_i}w_{ij},
    & j=i. \\
    \end{array}\right. 
\end{array} 
\end{equation}

Since the agents
are divided into leader group and follower group, the follower matrix $\Omega_f$ can be partitioned as
\begin{align}\label{aef}
 \Omega_f=[\begin{array}{ll}
    \Omega_{fl}  &  \Omega_{f\!f}
    \end{array}],
\end{align}
where $ \Omega_{fl} \in \mathbb{R}^{(n\!-\!m)\times m}$, $\Omega_{f\!f} \in \mathbb{R}^{(n\!-\!m)\times (n\!-\!m)}$. Then, \eqref{form} becomes
\begin{equation}\label{2dif}
   (\Omega_{fl} \otimes I_d)r_l +    (\Omega_{f\!f} \otimes I_d)r_f = \mathbf{0}.
\end{equation}

If the matrix $\Omega_{f\!f}$ is nonsingular, it yields  from \eqref{2dif} that
\begin{equation}
 r_f = -(\Omega_{f\!f}^{-1}\Omega_{fl}\otimes I_d) r_l.   
\end{equation}

Hence, $r_f$ can be uniquely determined by $r_l$ if the matrix $\Omega_{f\!f}$ is nonsingular. 
The constant edge weights such as $w_{ij}$ in \eqref{zer} will then be used to describe the constant geometric relationship among the time-varying target positions of follower $i$ and its neighbors.

\begin{defn}
A leader-follower nominal formation $(\mathcal{G}, r)$ is said to be localizable if the nominal positions of the  followers $r_f$ can be uniquely determined by those of the leaders $r_l$. A leader-follower time-varying target formation $(\mathcal{G}, p^*(t))$ is said to be localizable if the target positions of the  followers $p^*_f(t)$ can be uniquely determined by those of the leaders $p^*_l(t)$.
\end{defn}

\begin{theorem}\label{lsl2}
A  leader-follower nominal formation $(\mathcal{G}, r)$ over a multi-layer $d\!+\!1$-rooted graph is localizable if
for each follower $i (i=m\!+\!1, \cdots, n)$, the nominal positions of its $d\!+\!1$ neighbors $ r_{j_0},r_{j_1}, \cdots, r_{j_d}$ are not on a hyperplane in $\mathbb{R}^{d}$.
\end{theorem}
\begin{proof}
Based on a multi-layer $d\!+\!1$-rooted graph, we can know that $\Omega_{f\!f}$ in \eqref{2dif} is a lower triangular matrix, i.e.,
\begin{equation}\label{up2}
  \Omega_{f\!f} \!=\!
  \left[\begin{array}{lllll}
   \ \   {w_{(m\!+\!1)(m\!+\!1)}}   &  &  &  & \ \ 0   \\
    -w_{(m\!+\!2)(m\!+\!1)} &  \ \  w_{(m\!+\!2)(m\!+\!2)}    &    \\
   \ \ \ \ \ \  \vdots &  \ \ \ \ \ \ \vdots &    \ddots & &  \\
     -w_{n(m\!+\!1)} & -w_{n(m\!+\!2)}  & \cdots & & w_{nn}
    \end{array}\right].
\end{equation}

It is clear that the matrix $\Omega_{f\!f}$ is nonsingular if and only if its diagonal entries are non-zero, i.e.,
\begin{equation}\label{up3}
 w_{ii} \neq 0, \  i=m+1, \cdots, n. 
\end{equation}

From Lemma \ref{ls2} and \eqref{we}, we can conclude that $w_{ii} \neq 0, i=m\!+\!1, \cdots, n$ if for each follower $i (i=m\!+\!1, \cdots, n)$, the nominal positions of its $d\!+\!1$ neighbors $ r_{j_0},r_{j_1}, \cdots, r_{j_d}$
are not on a hyperplane in $\mathbb{R}^{d}$. Then, the conclusion follows.
\end{proof}

A simple example of 2-D localizable nominal formation over a multi-layer $3$-rooted graph is given in Fig. \ref{2d}, where the nominal positions of the leaders $r_l=[r_1^T ,r_2^T, r_3^T]^T$ and the nominal positions of the followers 
$r_f= [r_4^T,r_5^T, r_6^T]^T$ are
\begin{equation}
\begin{array}{ll}
     &  r_1 = [1,0]^T, \ \ r_2 = [0,1]^T, \ \
      r_3 = [0,-1]^T, \\
      & r_4 = [-1,1]^T, \ \ r_5 = [-1,-1]^T, \ \ r_6 = [-2, 0]^T.
\end{array}
\end{equation}

The corresponding matrices $ \Omega_{fl}$ and $ \Omega_{f\!f}$ are calculated as 
\begin{equation}\label{o4}
\Omega_{fl} \!=\! \left[\begin{array}{lll}
   -2 & 3 &  \ \  1   \\
  \ \ 0 & 1 & -1 \\
   \ \  0 & 0 &   -2
    \end{array}\right], \ \  
     \Omega_{f\!f} = \left[\begin{array}{lll}
    -2 &  0 &  \ \ 0 \\
    -1 &  1 &  \ \ 0 \\
   \ \ 1 &  3 & -2
    \end{array}\right]. 
\end{equation}

It can be verified that  $r_f=-(\Omega_{f\!f}^{-1}\Omega_{fl}\otimes I_2) r_l$, i.e., the 2-D nominal formation $(\mathcal{G}, r)$ shown in Fig. \ref{2d} is localizable.

\begin{lemma}\label{ta2}
A leader-follower time-varying target formation $(\mathcal{G}, p^*(t))$ in \eqref{ti} is  localizable if 
its constant nominal formation $(\mathcal{G}, r)$ is localizable.
\end{lemma}

\begin{proof}
We can know from \eqref{zer} that
\begin{equation}\label{tt2}
    \Omega_{fl} \cdot \mathbf{1}_{m} + \Omega_{f\!f} \cdot \mathbf{1}_{n\!-\!m} = \mathbf{0}.
\end{equation}

Combining\eqref{ti}, \eqref{2dif}, and \eqref{tt2}, it has
\begin{equation}\label{ty1}
    \begin{array}{ll}
         & (\Omega_{f\!f} \otimes I_d )p^*_f(t) + (\Omega_{fl} \otimes I_d) p^*_l(t) \\
         & = \beta(t) [(\Omega_{fl} \otimes Q(t))r_l \!+\! (\Omega_{f\!f} \otimes Q(t))r_f] \\
         & \ \ \ + [\Omega_{fl} \cdot \mathbf{1}_{m} + \Omega_{f\!f} \cdot \mathbf{1}_{n\!-\!m} ]\otimes \delta(t) \\
        & \!= \!\mathbf{0}.
    \end{array}
\end{equation}

If the nominal formation $(\mathcal{G}, r)$ is localizable, i.e., the matrix $\Omega_{f\!f}$ in \eqref{ty1} is nonsingular, it has
\begin{equation}\label{re}
p_f^*(t) = -(\Omega_{f\!f}^{-1}\Omega_{fl}\otimes I_d) p_l^*(t).
\end{equation}

Hence, the time-varying target formation $(\mathcal{G}, p^*(t))$ is localizable.

\end{proof}

\begin{remark}
It is clear from \eqref{ty1} that the constant edge weights such as $w_{ij}$ are invariant to translation $\delta(t)$, rotation $Q(t)$, and scaling $\beta(t)$ of the nominal formation. Hence, the followers can use this invariance property to achieve formation maneuver control without the need of knowing the maneuver parameters $\delta(t),Q(t),\beta(t)$ shown in Section \ref{single}, i.e., the followers do not need know their target positions.
\end{remark}

\section{Maneuvering Mode and Maintaining Mode}\label{mmm}

\begin{assumption}\label{csr}
For each follower $i \! \in 
 \! \mathcal{V}_f$  in $\mathbb{R}^d$, among the neighbors of follower $i$, there exist 
 $d\!+\!1$ neighbors ${j_0}, \cdots, {j_{d}} \! \in 
\! \mathcal{N}_i$ whose nominal positions are not on a hyperplane. The target configuration $p^*(t)$ in \eqref{ti} is first order differentiable.
\end{assumption}

Note that under Assumption \ref{csr}, we can know from Theorem \ref{lsl2} and Lemma \ref{ta2} that the time-varying target formation $(\mathcal{G}, p^*(t))$ in \eqref{ti} is localizable.  
Then, the control objectives \eqref{con1l} and  \eqref{con1} become
\begin{equation}\label{pro1}
\begin{array}{ll}
    &  \left\{ \! \begin{array}{lll} 
    \lim\limits_{t \rightarrow \infty} (p_i(t)-p^*_i(t)) = \mathbf{0}, \ \ i \in \mathcal{V}_l,  \\
       \lim\limits_{t \rightarrow \infty}(\hat p_i(t)-p_i(t)) = \mathbf{0}, \ \  i \in \mathcal{V}_f, \\
       \lim\limits_{t \rightarrow \infty} (p_i(t) - \sum\limits_{j \in \mathcal{N}_i} \frac{w_{ij}}{w_{ii}}   \hat {p}_{j}(t) )= \! \mathbf{0}, \ i \in \mathcal{V}_f,
    \end{array}\right. 
\end{array} 
\end{equation}
where $\mathcal{N}_i$ is given in \eqref{neighborset} and $\hat p_j = p_j, j \in \mathcal{V}_l$. 

\begin{remark}
From \eqref{ty1}, it has $p^*_f(t) \!+ \!(\Omega_{f\!f}^{-1}\Omega_{fl} \otimes I_d)p^*_l(t) \!=\!  \mathbf{0}$. The leaders know their own positions, i.e., $\hat p_l(t) = p_l(t)$. Based on the first and second subequations of \eqref{pro1}, the third subequation of \eqref{pro1} is equivalent to $\lim\limits_{t \rightarrow \infty} [p_f(t) + (\Omega_{f\!f}^{-1}\Omega_{fl} \otimes I_d)p^*_l(t) ] \!=\!  \mathbf{0}$, i.e., $\lim\limits_{t \rightarrow \infty} [p_i(t) - p^*_i(t) ] \!=\!  \mathbf{0},  i \! \in \! \mathcal{V}_f$.
\end{remark}

Based on \eqref{pro1}, 
define the tracking error $\bar e_i$ and position estimation error $\hat e_i$ of follower $i \in \mathcal{V}_f$ as
\begin{align}
         &     \bar e_i(t) = p_i(t)  -\!\sum\limits_{j \in \mathcal{N}_i}\frac{w_{ij}}{w_{ii}}{\hat p}_{j}(t). \ \label{tra1} \\
         & \hat e_i(t) = \hat p_i(t) -p_i(t).\label{tra2}
\end{align}

As shown in \eqref{root}, at any time instant $t$, follower $i$ and its $d\!+\!1$ neighbors ${j_0}, \cdots, {j_{d}} \in \mathcal{N}_i$ can form a displacement constraint, i.e.,
\begin{equation}\label{root1}
 \sum\limits_{j \in \mathcal{N}_i}h_{ij}(t) e_{ij}(t) \!=\! \mathbf{0} \Rightarrow  h_{ii}(t)p_i(t) \!=\! \sum\limits_{j \in \mathcal{N}_i}h_{ij}(t) p_{j}(t),
\end{equation}
where $h_{ii}(t) \!=\! \sum\limits_{j \in \mathcal{N}_i}h_{ij}(t)$, and the displacement parameters $h_{ij}(t), j \in \mathcal{N}_i$ can be calculated by the relative positions, bearings, distances, angles, ratio-of-distances, or their combination among agent $i$ and its $d\!+\!1$ neighbors ${j_0}, \cdots, {j_{d}} \in \mathcal{N}_i$ (the details can be found in Appendix and \cite{fang20203}). 
Under Assumption \ref{csr},
we can conclude from  \eqref{ti} and Lemma \ref{ls2}  that $h_{ii}(t) \! \neq \! 0$ if 
its $d\!+\!1$ neighbors ${j_0}, \cdots, {j_{d}} \in \mathcal{N}_i$ reach their target positions, and then \eqref{root1} can be rewritten as
\begin{equation}\label{we2}
p_i(t) \!=\! \sum\limits_{j \in \mathcal{N}_i} \frac{h_{ij}(t)}{h_{ii}(t)}p_{j}(t).
\end{equation}

Although we can guarantee the localizability of the target formation 
$(\mathcal{G}, p^*(t))$, there is no guarantee on the localizability of follower agents during the transient before all agents reach the target formation. To tackle this problem, two operation modes (maneuvering mode and maintaining mode) are designed for each follower. A remarkable advantage of the two operation modes is that we do not require that all the followers be always localizable by their neighbors before all agents reach the target formation.

\begin{defn}\label{ma1}
Follower $i \! \in \! \mathcal{V}_f$ is in the maneuvering mode at time instant $t$ if 
its $d\!+\!1$ neighbors ${j_0}, \cdots, {j_{d}} \in \mathcal{N}_i$ reach their target positions. 
\end{defn}

\begin{defn}\label{ma2}
Follower $i \! \in \! \mathcal{V}_f$ is in the maintaining mode at time instant $t$ if at least one of its $d\!+\!1$ neighbors ${j_0}, \cdots, {j_{d}} \in \mathcal{N}_i$ does not reach its target position.
\end{defn}

The immediate question in Definition \ref{ma1} and Definition \ref{ma2} is how each agent knows whether it has arrived at its target position? Since the leaders know their own positions and target positions, each leader $i \! \in \! \mathcal{V}_l$ can know whether it has arrived at its target position. 
Next, we will introduce how each follower $i \! \in \! \mathcal{V}_f$ knows whether 
it achieves self-localization and has arrived at its target position. 
Although the followers do not know their target positions, 
each follower $i \! \in \! \mathcal{V}_f$ under Assumption \ref{csr} can know that its task is completed 
if and only if
\begin{equation}\label{get}
\begin{array}{ll}
    &  \left\{ \! \begin{array}{lll} 
       \sum\limits_{j \in \mathcal{N}_i} (\frac{h_{ij}(t)}{h_{ii}(t)} \!-\! \frac{w_{ij}}{w_{ii}}) \hat p_{j}(t) = \mathbf{0},  \\
 \hat p_i(t) - \sum\limits_{j \in \mathcal{N}_i} \frac{h_{ij}(t)}{h_{ii}(t)} \hat p_{j}(t) \!=\! 
\mathbf{0},  
       \\
       \hat p_{j}(t) = p_{j}(t)= p^*_{j}(t), \ j \in \mathcal{N}_i.
    \end{array}\right. 
\end{array} 
\end{equation}

\begin{remark}
Each follower $i \! \in \! \mathcal{V}_f$ needs to communicate with its neighbor $j \! \in \! \mathcal{N}_i$ to obtain its neighbor's position estimate $\hat p_j$ and to know whether its neighbor  $j \! \in \! \mathcal{N}_i$ satisfies the third subequation of \eqref{get}.
The position estimates $\hat p_i, \hat p_j$ and control inputs $u_i, u_j$ of each follower $i \! \in \! \mathcal{V}_f$ and its neighbor $j \! \in \! \mathcal{N}_i$ will be designed in Section \ref{single} 
to guarantee that the condition \eqref{get} holds. 
If the third subequation of
\eqref{get} holds, we can know from \eqref{we2} that
the first and second subequations of
\eqref{get} are equivalent to
\begin{equation}\label{kj}
    \begin{array}{ll}
         & \hspace{-0.4cm} \sum\limits_{j \in \mathcal{N}_i} \! \!(\frac{h_{ij}(t)}{h_{ii}(t)} \!-\! \frac{w_{ij}}{w_{ii}}) \hat p_{j}(t) \!=\! p_i(t) \!-\!  \frac{w_{ij}}{w_{ii}} p^*_{j}(t) \!=\! p_i(t) \!-\! p^*_{i}(t) \!=\! \mathbf{0}, \\
         & \hspace{-0.4cm}  \hat p_i(t)  - \! \sum\limits_{j \in \mathcal{N}_i} \! \! \frac{h_{ij}(t)}{h_{ii}(t)} \hat p_{j}(t) = \hat p_i(t) - p_i(t)  =
\mathbf{0}.
\end{array}
\end{equation}

Thus, we can know from \eqref{kj} that 
the condition \eqref{get} is equivalent to $\hat p_i(t) \!=\! p_i(t) \!=\! p^*_i(t), i \! \in \! \mathcal{V}_f$. That is, each follower $i \! \in \! \mathcal{V}_f$ can know whether it has arrived at its target position by \eqref{get}. If the third subequation of \eqref{get} does not hold, follower $i$ 
cannot achieve self-localization or know whether it has arrived at its target position. 
In the multi-layer $d\!+\!1$-rooted graphs, the followers will sequentially make their own conditions \eqref{get} hold under the proposed controllers shown in Section \ref{single}, i.e., the followers will sequentially reach their target positions over the multi-layer $d\!+\!1$-rooted graphs.
\end{remark}

From the above analysis, each leader or follower can know whether 
it has arrived at its target position. Then, 
each follower $i \! \in \! \mathcal{V}_f$ can know which mode it belongs to according to Definition \ref{ma1} and Definition \ref{ma2} by communicating with its neighbors $j \! \in \! \mathcal{N}_i$.

\section{Control Design for the Agents}\label{single}

\subsection{Control Design for the Leaders}\label{sin1}

We first define a continuous function $\text{sig}(\cdot)$ as
\begin{equation}
  \text{sig}^{l}(x) \!=\! [\text{sgn}(x_1)|x_1|^{l}, \cdots, \text{sgn}(x_d)|x_d|^{l}]^T, 
\end{equation}
where $x \!=\! [x_1, \cdots, x_d]^T \! \in \! \mathbb{R}^d, l\!>\!0$, and $\text{sgn}(\cdot)$ is the signum function defined component-wise. Then, it has
\begin{equation}
\begin{array}{ll}
     &    x^T \text{sig}^{l}(x)=  |x_1|^{l\!+\!1}+ \cdots +
    |x_d|^{l\!+\!1} \\
     & \ \ \ \ \ \ \ \ \ \ \ \  = (x_1^2)^{\frac{l\!+\!1}{2}}+ \cdots +
    (x_d^2)^{\frac{l\!+\!1}{2}}.
\end{array}
\end{equation}

If $0<l<1$ \cite{hardy1952inequalities}, it has
\begin{equation}\label{base1}
\begin{array}{ll}
     &    x^T \text{sig}^{l}(x)=  (x_1^2)^{\frac{l\!+\!1}{2}}+ \cdots +
    (x_d^2)^{\frac{l\!+\!1}{2}} \\
     & \ \ \ \ \ \ \ \ \ \ \ \  \ge  (x_1^2+\cdots+x_d^2)^{\frac{l\!+\!1}{2}} = \|x\|_2^{l\!+\!1}.
\end{array}
\end{equation}

Let $\tilde e_i=p_i-p^*_i$ be the tracking error of leader $i$.
To achieve its control objective \eqref{con1l},  the control protocol of leader $i \in \mathcal{V}_l$ is given by
\begin{equation}\label{controlear}
    u_i = - a_1 \tilde e_i -a_2\text{sig}^{a_3}(\tilde e_i) + \dot p^*_i, 
\end{equation}
where $a_1, a_2>0, 0<a_3<1$ are positive control gains.

\begin{lemma}\label{leaderlemma1}
Under Assumption \ref{csr}, the leaders  will drive their tracking errors to zero in finite time under controller \eqref{controlear}.
\end{lemma}
\begin{proof}
Consider a Lyapunov function candidate $V_1 \!=\! \sum\limits_{i=1}^{m}{\tilde e}^T_i{\tilde e}_i$. It has
\begin{equation}
\begin{array}{ll}
\begin{array}{ll}
     & \dot V_1 \!=\! -2a_1\sum\limits_{i=1}^{m}{\tilde e}^T_i{\tilde e}_i \!-\! 2a_2\sum\limits_{i=1}^{m}\tilde e_i^T\text{sig}^{a_3}(\tilde e_i) \\
     & \ \ \ \ \!\le\! -2a_2\sum\limits_{i=1}^{m}\|\tilde e_i\|_2^{a_3\!+\!1} \! \le \! -2a_2V_1^{\frac{a_3\!+\!1}{2}}.
\end{array}
\end{array}
\end{equation}

From Lemma $1$ of \cite{hong2001output}, the tracking errors of the leaders will converge to zero in finite time, i.e., there exists a finite time $T, 0<T<\infty$ such that $\tilde e_i(t)=0, t\ge T, i \in \mathcal{V}_l$, where $T$ is determined by the initial tracking errors $\sum\limits_{i=1}^{m}{\tilde e}^T_i(0){\tilde e}_i(0)$.
\end{proof}

\subsection{Control Design for the Followers under Maintaining Mode}\label{smtm}

From Section \ref{mmm}, if follower $i$ is in the maintaining mode, it may not be localized by its neighbors. 
Hence, 
it should keep its formation error bounded before switching to the maneuvering mode. To keep its formation error bounded, 
the controller of follower $i \in \mathcal{V}_f$  is designed as
\begin{equation}\label{int1}
\begin{array}{ll}
    &  \left\{ \! \begin{array}{lll} 
    u_i \!=\! - a_1\hat p_i \!+\! \sum\limits_{j \in \mathcal{N}_i} \frac{w_{ij}}{w_{ii}}(a_1{\hat p}_{j} \!+\!  {\hat v}_{j}), \\
      \dot {\hat p}_i \!=\! - 2a_1\hat p_i \!+\! \sum\limits_{j \in \mathcal{N}_i} \frac{w_{ij}}{w_{ii}}(2a_1{\hat p}_{j} \!+\! {\hat v}_{j}),
    \end{array}\right. 
\end{array} 
\end{equation}
where $a_1 >0$ is a positive control gain, and ${\hat v}_{j}  \!=\! \dot {\hat p}_j$ is the estimated velocity.

\begin{theorem}\label{mtml}  Follower $i \in \mathcal{V}_f$ under the maintaining mode will keep its formation error bounded under controller \eqref{int1}.
\end{theorem}
\begin{proof}
Under controller \eqref{int1}, it has 
\begin{equation}\label{mtm1}
    \dot e_i \!=\! -D_ae_i, \ i \!=\! m\!+\!1, \cdots, n, 
\end{equation}
where  $e_i=(\bar e_i^T, \hat e_i^T)^T$ is the formation error of follower $i$ and $D_a \!=\!   \left[\begin{array}{ll}
   a_1 & a_1  \\
    a_1 & a_1
    \end{array}\right] \! \otimes \! I_d. $ The formation error $e_i$ of follower $i$ consists of the tracking error $\bar e_i$ \eqref{tra1} and position estimation error $\hat e_i$ \eqref{tra2}. 
Note that the matrix $D_a$ is positive semidefinite. Consider a Lyapunov function candidate $V_2=\frac{1}{2}e^T_ie_i$. It has $ \dot{{V}}_2 \!=\! -e_i^TD_ae_i \! \le  \! 0.$
Hence, we obtain $\|e_i(t)\|_2 \! \le \!  \|e_i(0)\|_2, t \! \ge \! 0$, i.e., the formation error $e_i(t)$ of follower $i$ will be no more than its initial formation error $\|e_i(0)\|_2$.

\end{proof}

\subsection{Control Design for the Followers under Maneuvering Mode}\label{smvm}

From Section \ref{mmm}, if follower $i$ is in the maneuvering mode, it must be localized by its neighbors. Hence, 
if follower $i$ is in the maneuvering mode, it should achieve self-localization 
and get to its target position. To achieve self-localization 
and get to its target position,
the control protocol of follower $i \in \mathcal{V}_f$ is given by
\begin{equation}\label{int2}
\begin{array}{ll}
    &  \left\{ \! \begin{array}{lll} 
    u_i \!=\!  \eta_i \!+\! \sum\limits_{j \in \mathcal{N}_i} \frac{w_{ij}}{w_{ii}}  {\hat v}_{j} \!+\! \text{sig}^{a_3}(\eta_i ), \\
      \dot {\hat p}_i \!=\! 2\eta_i \!+\! \sigma_i \!+\!\sum\limits_{j \in \mathcal{N}_i} \frac{w_{ij}}{w_{ii}}  {\hat v}_{j}  \!+\! \text{sig}^{a_3}(\eta_i )  \!+\! \text{sig}^{a_3}(\eta_i \!+\! \sigma_i), \\
      \eta_i = - a_2\hat p_i \!+\! a_2\sum\limits_{j \in \mathcal{N}_i} \frac{w_{ij}}{w_{ii}}{\hat p}_{j}, \\
      \sigma_i = -a_4\hat p_i\!+\!a_4\sum\limits_{j \in \mathcal{N}_i} \frac{h_{ij}}{h_{ii}}\hat p_{j},
    \end{array}\right. 
\end{array} 
\end{equation}
where  $a_2, a_4 > 0, 0 < a_3 <1$ are the positive control gains.

\begin{theorem}\label{mvml}
Under Assumption \ref{csr},  follower $i \in \mathcal{V}_f$ under the maneuvering mode will drive its formation error to zero in finite time under controller \eqref{int2}.
\end{theorem}
\begin{proof}
Let $e_i=(\bar e_i^T, \hat e_i^T)^T$ be the formation error of follower $i \! \in \! \mathcal{V}_f$, where $\bar e_i, \hat e_i$ are given in \eqref{tra1} and \eqref{tra2}. Under controller \eqref{int2}, it has
\begin{equation}\label{int4}
     \dot e_i \!=\! -D_be_i - \text{sig}^{a_3}(D_be_i), \ i \!=\! m\!+\!1, \cdots, n, 
\end{equation}
where $0
\!<\! \frac{a_3\!+\!1}{2} \!<\! 1$ and $D_b \!=\!  \left[\begin{array}{ll}
   a_2 & a_2  \\
    a_2 & a_2\!+\!a_4
    \end{array}\right] \! \otimes \! I_d$. Note that the matrix $D_b$ is positive definite.
Consider a Lyapunov function candidate $V_3=\frac{1}{2}e^T_iD_be_i$.
Based on \eqref{base1} and \eqref{int4}, it has
\begin{equation}\label{int5}
\begin{array}{ll}
     & \dot{{V}}_3 = -e_i^TD_b^2e_i - (D_be_i)^T\text{sig}^{a_3}(D_be_i) \le -\|D_be_i\|_2^{a_3\!+\!1}.
\end{array}
\end{equation}

Note that $V_3 \le \frac{1}{2}\lambda_{\max}(D_b)\|e_i\|_2^2$. Then, \eqref{int5} becomes
\begin{equation}\label{int6}
\begin{array}{ll}
    & \dot{{V}}_3  \le -\|D_be_i\|_2^{a_3\!+\!1}=- (e_i^TD_b^2e_i)^{\frac{a_3\!+\!1}{2}} \\
     & \ \ \ \ \le -(e_i^T\lambda^2_{\min}(D_b)e_i)^{\frac{a_3\!+\!1}{2}} = -\lambda^{a_3\!+\!1}_{\min}(D_b)(\|e_i\|_2^{2})^{\frac{a_3\!+\!1}{2}} \\ 
      & \ \ \ \ \le  -(\frac{2\lambda^{2}_{\min}(D_b)}{\lambda_{\max}(D_b)})^{\frac{a_3\!+\!1}{2}} V_3^{\frac{a_3\!+\!1}{2}}.
\end{array}
\end{equation}

From Lemma $1$ of \cite{hong2001output}, the formation error $e_i(t)$ of follower $i$ will converge to zero in finite time, i.e., there exists a finite time $T, 0<T<\infty$ such that $e_i(t)=0, t\ge T$.
\end{proof}

\begin{figure}[t]
\centering
\includegraphics[width=1\linewidth]{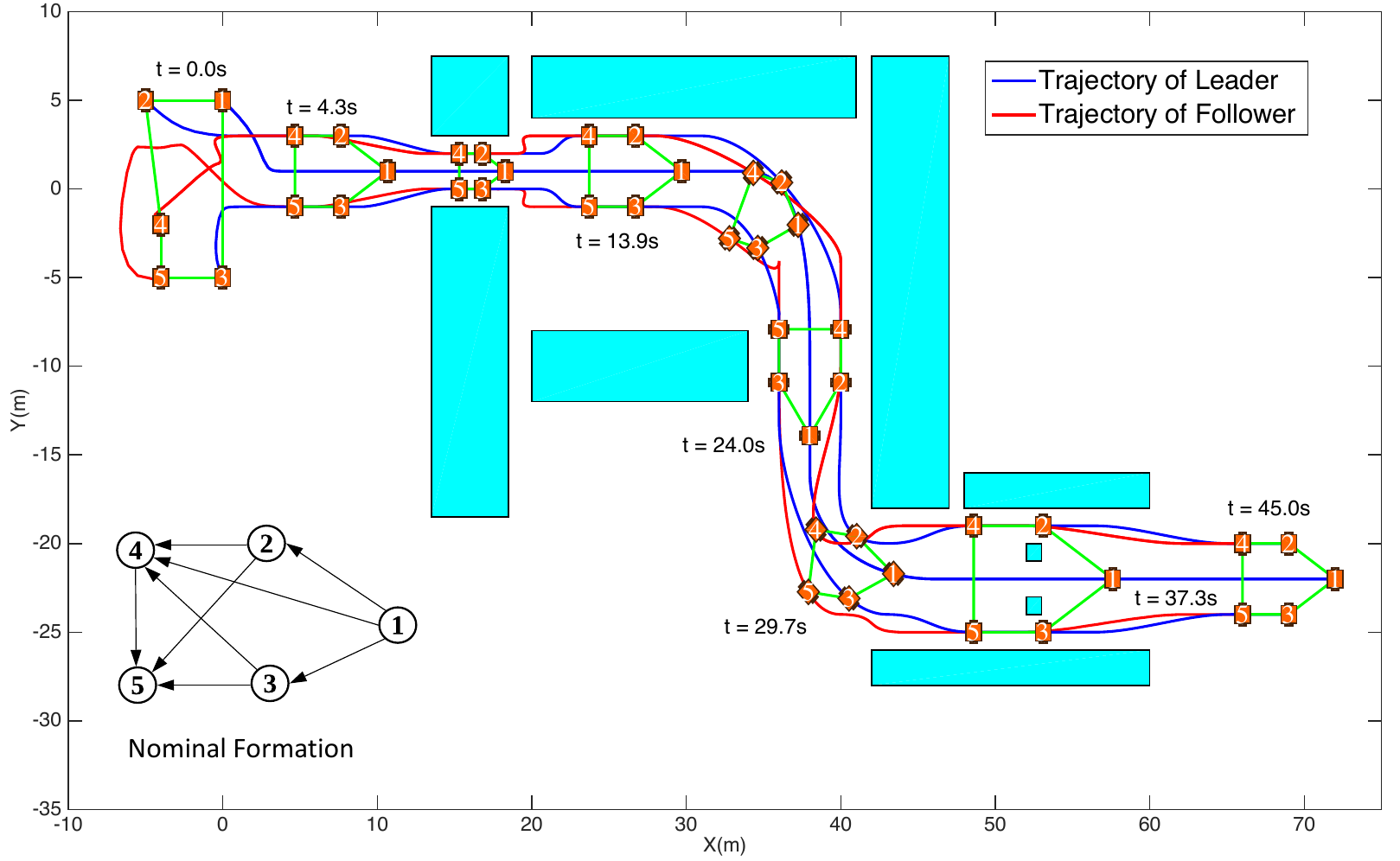}
\caption{Trajectories of multi-agent system}
\label{simu1}
\end{figure}

From Theorem \ref{mtml} and Theorem \ref{mvml}, we have the following conclusion.

\begin{theorem}\label{them1}
Under Assumption \ref{csr},
the agents over a multi-layer $d \!+\! 1$-rooted graph will converge to their target positions globally in finite time under the controllers \eqref{controlear}, \eqref{int1}, and \eqref{int2}.
\end{theorem}
\begin{proof}

For any given initial positions of the leader group $\mathcal{V}_l=\{1, \cdots, m \}$ and follower group $\mathcal{V}_f\!=\! \{m\!+\!1, \cdots, n\}$,
from Lemma \ref{leaderlemma1} and under controller \eqref{controlear}, the leaders will converge to their target positions in finite time $T_1>0$, i.e.,
\begin{equation}
 \|\tilde e_i(t)\|_2 \!=\! 0, \ i \!=\! 1, \cdots, m,  \ t \ge T_1.    
\end{equation}

During the time interval $[0, T_1)$, we can know from Definition  \ref{ma2} that all followers will be in the maintaining modes. From Theorem \ref{mtml}, all followers will keep their formation errors bounded, i.e.,
\begin{equation}
\begin{array}{ll}
     \|e_i(t)\|_2 \le \|e_i(0)\|_2, \ i \!=\! m\!+\!1, \cdots, n,  \ 0 \le t  < T_1. 
\end{array} 
\end{equation}

As shown in \eqref{up2}, based on a multi-layer $d \!+\! 1$-rooted graph \eqref{neighborset}, 
the first follower ${m\!+\!1}$ has
$d\!+\!1$ neighboring leaders. Then, the first follower ${m\!+\!1}$ will switch to the maneuvering mode at time instant $T_1$ because its $d\!+\!1$ neighboring leaders get to their target positions at time instant $T_1$. From Theorem \ref{mvml}, the first follower ${m\!+\!1}$ will 
drive its formation error to zero in finite time $T_2 \! \ge \! 0$, i.e., the condition \eqref{get} holds and $ e_{m\!+\!1}(t) \!=\!0, t  \ge T_1 \!+\! T_2.$
During the time interval $[T_1, T_1\!+\!T_2)$, the rest followers $i = m\!+\!2, \cdots, n$ will still be in the maintaining modes to keep their formation errors bounded, i.e.,
\begin{equation}
\begin{array}{ll}
     \|e_i(t)\|_2 \le \|e_i(0)\|_2, \ i \!=\! m\!+\!2, \cdots, n,  \ 0 \le t  < T_1 \!+\! T_2. 
\end{array} 
\end{equation}

Then, the rest followers $i \!=\! m\!+\!2, \cdots, n$ will also switch to the maneuvering modes and converge to their target positions sequentially in finite time by using the similar argument. Hence, all agents will converge to their target positions globally in finite time $T=\sum\limits_{i=1}^{n\!-\!m\!+\!1}T_i$ under controllers \eqref{int1} and \eqref{int2}.

\end{proof}

\begin{remark}
Note that the work in \cite{mehdifar20222} is applicable to unaligned local coordinate frames. 
One possible way to extend the proposed method to unaligned local coordinate frames is to design an orientation estimation protocol for each follower to estimate the orientation of its local coordinate frame by combining other types of measurements such as relative orientation measurements \cite{li2019globally}.
In theorem \ref{them1}, only the first leader $p_{m\!+\!1}$ has $d\!+\!1$ leaders as neighbors.
The case that no follower has $d\!+\!1$ leaders as neighbors yet still localizable and the extension of the proposed method are given in Appendix.      
\end{remark}

\begin{figure}[t]
\centering
\includegraphics[width=0.8\linewidth]{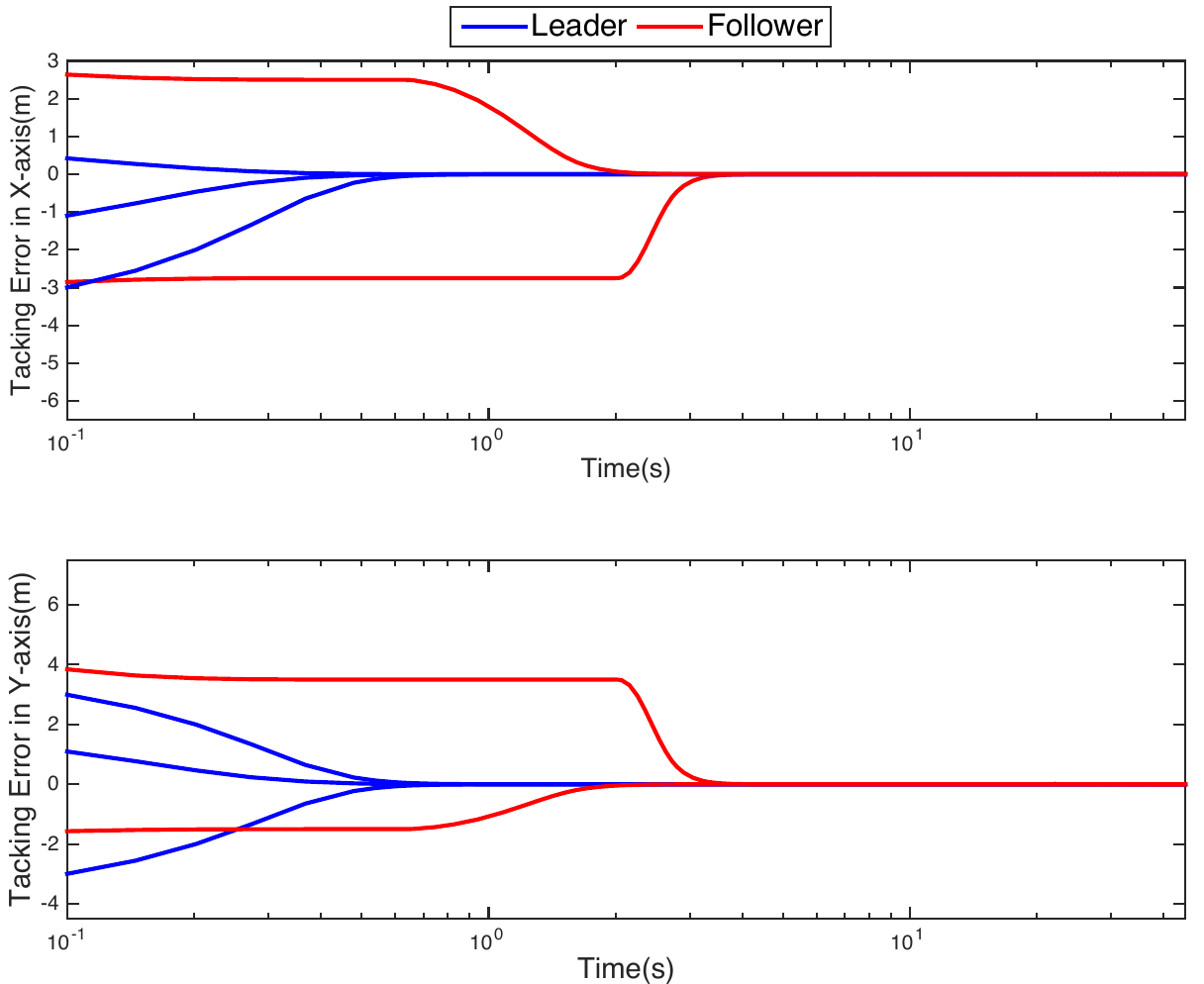}
\caption{Tracking errors of the agents.}
\label{simu2}
\end{figure}

\section{Simulation}\label{simul}

In this section, we present a 2-D formation with three leaders $V_l=\{1,2,3\}$ and two followers $V_f=\{4,5\}$. The leaders $1,2,3$ have access to their positions.
The follower $4$ can measure angle, while the follower $5$ can measure distance.
The localizable nominal configuration $r= [r_1^T ,r_2^T, r_3^T,r_4^T,r_5^T]^T$ in 2-D space is given by
\begin{equation}
\begin{array}{ll}
     & r_1 = [2,1]^T, \ \ r_2 = [-1,3]^T, \ \ r_3 = [-1,-1]^T,  \\
     & r_4 = [-4,3]^T, \ \ r_5 = [-4,-1]^T.
\end{array}
\end{equation}

The matrices $\Omega_{fl}$ and $\Omega_{f\!f}$ are calculated as
\begin{equation}\label{sf}
    \Omega_{fl} = \left[\begin{array}{lll}
    2 & -3 &  -1 \\
    0 & \ \ 1 & -1
    \end{array}\right],  \ \ 
     \Omega_{f\!f} = \left[\begin{array}{ll}
    \ \ 2 & 0 \\
    -1 & 1
    \end{array}\right].
\end{equation}

It is clear that $r_f \! = \! -(\Omega_{f\!f}^{-1}\Omega_{fl}\otimes I_2) r_l$, i.e., the nominal formation is localizable. The integrated distributed localization and formation maneuver control in 2-D space is 
shown in Fig. \ref{simu1}, where the multi-agent system passes through the narrow spaces and avoids obstacles by changing its scale, orientation, and rotation. For example, during the time interval $t \in (6, 10]$, 
the agents pass through the first narrow passage, where
the scaling, rotation, and translation maneuver parameters are designed as $ \beta(t) =\frac{1}{2},   Q(t) =I_2,  \delta(t)= [2t\!+\!1, \frac{1}{2}]^{\top}$. Then, the target configuration $p^*(t)$ in \eqref{ti} becomes
\begin{equation}
\begin{array}{ll}
     &  p^*(t)=\beta(t)[I_5 \otimes Q(t)]r+ {\mathbf{1}}_5 \otimes \delta(t) \\
     & \ \ \ \ \ \ \ \!=\! \frac{1}{2}r+{\mathbf{1}}_5 \otimes \left[\begin{array}{l}
    2t\!+\!1 \\
    \frac{1}{2}
    \end{array}\right], \ t \in (6, 10].  
\end{array}
\end{equation}

\begin{figure}[t]
\centering
\includegraphics[width=0.8\linewidth]{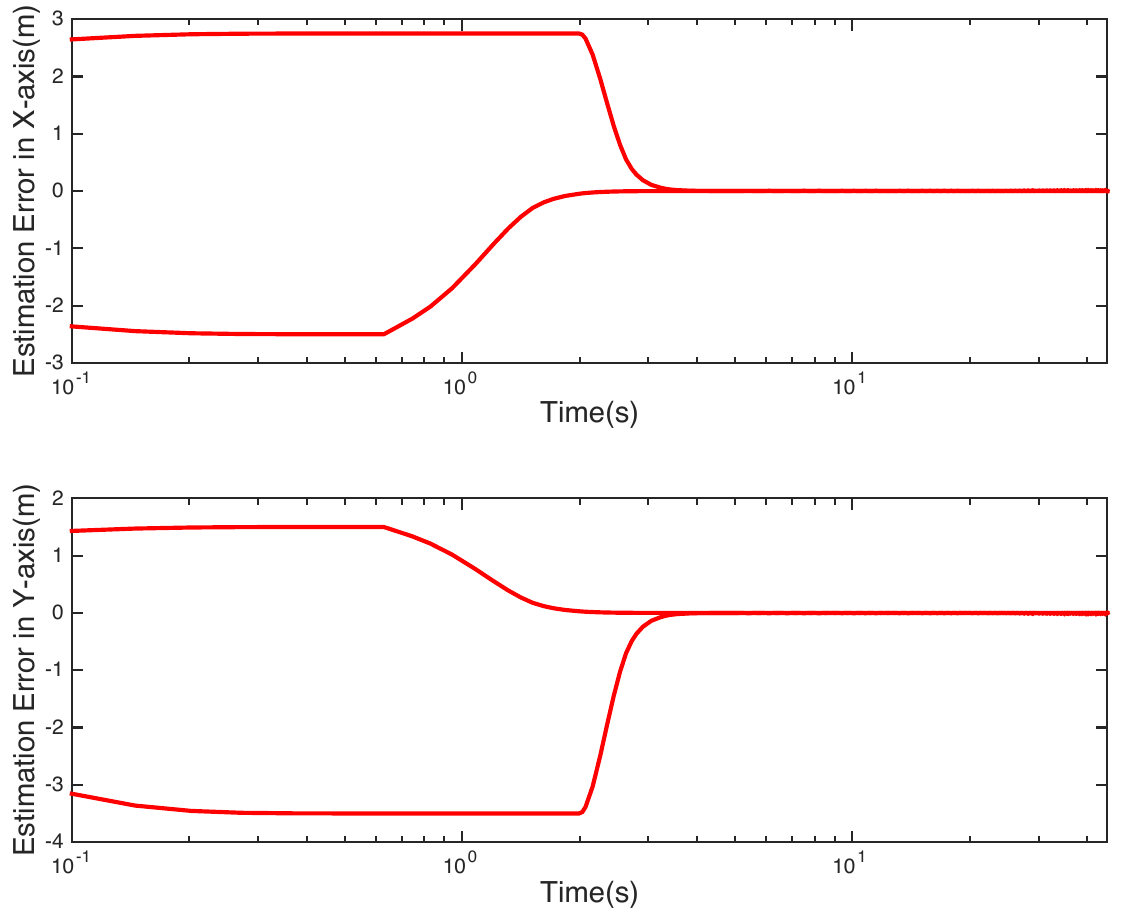}
\caption{Position estimation errors of the follower group.}
\label{simu3}
\end{figure}

The tracking errors and position estimation errors of the agents are given in Fig. \ref{simu2}, where the agents drive their tracking errors or position estimation errors to zero within finite time. It is shown in Fig. \ref{simu4} that the control inputs of the leader group and follower group are continuous.

\section{Conclusion}\label{concl}

This paper solves the integrated relative-measurement-based distributed  localization and formation maneuver control problem in $\mathbb{R}^d (d\ge2)$. The scale, rotation, and translation of the formation can be changed simultaneously by only controlling the positions of the leaders. In addition, the followers, which have no knowledge of their time-varying target positions, are not required to be localizable at all times and will converge to their target  positions globally and  sequentially in finite time. 

\begin{figure}[t]
\centering
\includegraphics[width=0.85\linewidth]{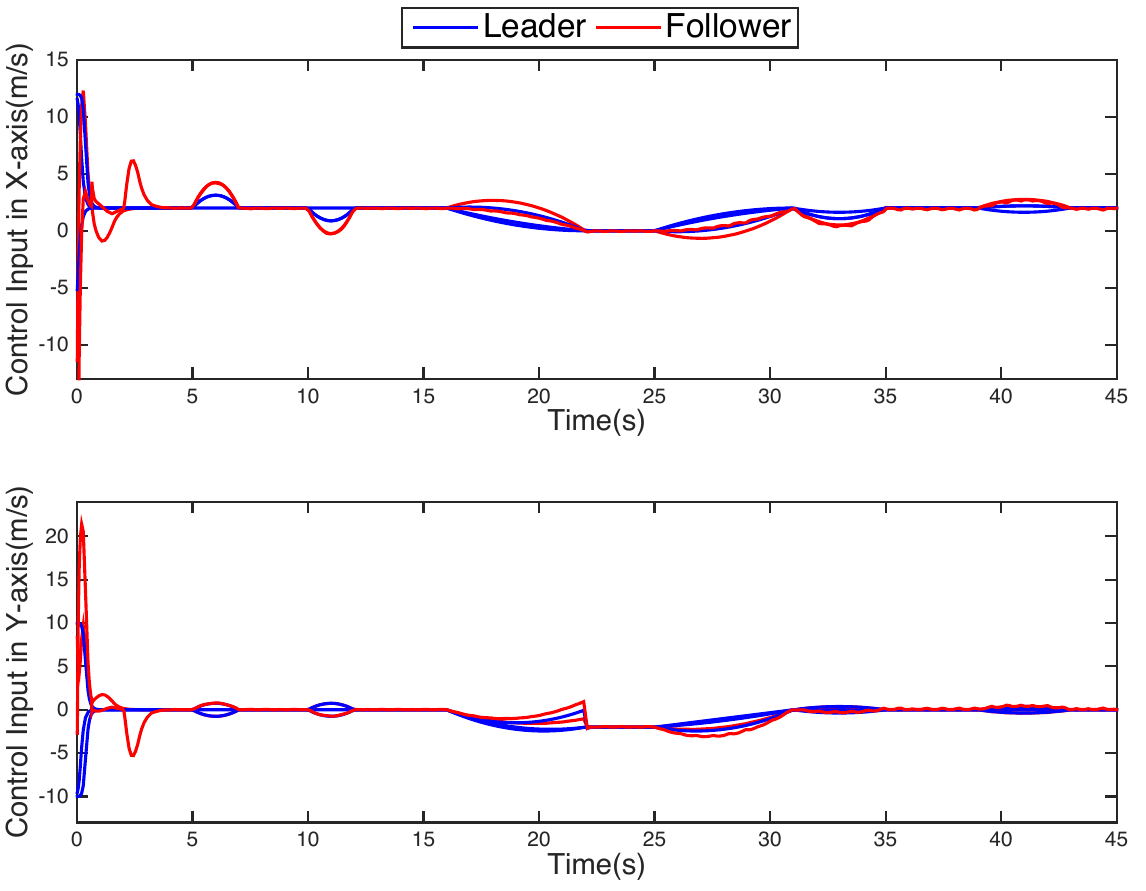}
\caption{Control inputs of the agents.}
\label{simu4}
\end{figure}

\section{Appendix}

Note that the relative positions, bearings, angles, and ratio-of-distances can be obtained by vision technology \cite{chen2011kalman,cao2019ratio,tron2016distributed}, which can then be regarded as communication-free measurements. In addition, an implicit assumption in this paper is that each follower $i$ and its neighbors $\mathcal{N}_i$ in \eqref{neighborset} are within their sensing ranges. For example, if agent $j$ is a neighbor of follower $i$,  the implicit assumption is that follower $i$ and the rest agents in $\mathcal{N}_i$ are within the sensing range of agent $j$. That is, agent $j \in \mathcal{N}_i$ equipped with vision sensors can obtain relative positions, bearings, angles, or ratio-of-distances with follower $i$ and the rest agents in $\mathcal{N}_i$ in a communication-free manner \cite{chen2011kalman,cao2019ratio,tron2016distributed}. Agent $j \in \mathcal{N}_i$ can transfer its measured information to follower $i$ through the directed edge $(i,j)$.

If follower $i$ is equipped with distance sensor, its followers can transfer their estimated positions to follower $i$. If the condition \eqref{get} holds, 
the neighbors of follower $i$ reach their target positions and achieve self-localization, then 
the distances among its neighbors $\mathcal{N}_i$ can be calculated based on the estimated positions of its neighbors. Hence, the displacement parameters $h_{ij_0},  \cdots,  h_{ij_d}$ in \eqref{root} can be calculated relative positions,  bearings, distances, angles, ratio-of-distances, or their combination among agent $i$ and its $d\!+\!1$ neighbors ${j_0}, \cdots, {j_{d}}$  over a multi-layer $d\!+\!1$-rooted graph shown in Section \ref{se1}-Section \ref{sef}. If the angles are obtained based on wireless technology, the edges among follower $i$ and its neighbors $\mathcal{N}_i$ should be revised accordingly.

\begin{lemma}
For follower $i$ and its $d\!+\!1$ neighbors ${j_0}, \cdots, j_{d}$, the displacement constraint \eqref{root1} is invariant to translations, rotations, and scalings of the configuration $\check{p}(t)=(p_i^T(t), p^T_{j_0}(t), \cdots, p^T_{j_d}(t))^T$.    
\end{lemma}

\begin{proof}
For the displacement constraint 
$\sum\limits_{k=0}^{d}h_{ij_k}(t) e_{ij_k}(t) \!=\! \mathbf{0}$ in \eqref{root1},
we have $\sum\limits_{k=0}^{d} \beta(t) Q(t) h_{ij_k}(t) e_{ij_k}(t) \!=\! \mathbf{0}$, where $\beta(t) \in \mathbb{R}$ is a scaling factor and $Q(t) \in {{SO}(d)}$ is a $d$-dimensional rotation matrix.  Hence, the displacement constraint is invariant to rotations and scalings of the configuration $\check{p}(t)$. Since the relative positions are invariant to translations, the displacement constraint \eqref{root1} is also invariant to translations of the configuration $\check{p}(t)$. 
\end{proof}

\begin{defn}
 Two configurations ${p}=(p^T_{1}, \cdots, p^T_{n})^T$ and ${q}=(q_1^T, \cdots, q^T_{n})^T$ are congruent if $\|p_i\!-\!p_j\|_2=\| q_i\!-\!q_j\|_2$ for any $i,j \in \mathcal{V}$. and similar if $ \|p_i\!-\!p_j\|_2 =  \kappa \|q_i\!-\!q_j\|_2,  \kappa >0$ for any $i,j \in \mathcal{V}$.    
\end{defn}

Next, we will introduce how to calculate displacement parameters $h_{ij_0}(t), \cdots, h_{ij_d}(t)$ in \eqref{root1} by only using  distances, angles,  ratio-of-distances, bearings, relative positions, or their mixture. In addition, we will introduce the relaxed graph condition, continuous velocity controller, and further comparison with the existing works. 

\subsection{Distance-based Displacement Constraint}\label{se1}

If the inter-agent distance measurements are available, we can obtain the distance matrix $M_d=[d_{ij}^2] \in \mathbb{R}^{(d\!+\!2) \times (d\!+\!2)}$ of follower $i$ and its $d\!+\!1$ neighbors ${j_0}, \cdots, {j_{d}}$, where $d_{ij}=\|p_j\!-\!p_i\|_2$ is the distance between agent $j$ and agent $i$ in $\mathbb{R}^d (d\ge 2)$.  For example, if  follower $p_i$ is in 2-D space, i.e., $d=2$, the distance matrix $M_2$ of follower $p_i$ and its three neighbors $p_{j_0}, p_{j_1}, p_{j_2}$ is given by
\begin{equation}\label{distance}
M_2 = 
    \left[ \!
\begin{array}{c c c c }
0 & d_{ij_0}^2   & d_{ij_1}^2 &  d_{ij_2}^2 \\
d_{j_0i}^2 & 0  & d_{j_0j_1}^2 & d_{j_0j_2}^2 \\
d_{j_1i}^2 & d_{j_1j_0}^2 & 0 & d_{j_1j_2}^2  \\
d_{j_2i}^2 & d_{j_2j_0}^2 & d_{j_2j_1}^2 & 0
\end{array}
\right].
\end{equation}

\begin{algorithm}
\caption{Obtaining $\check{q}=(q_i^T, q^T_{j_0}, \cdots, q^T_{j_d})^T$}
\label{disa}
\begin{algorithmic}[1]
\State Calculate the following matrices based on the distance matrix $M_d$. 
\begin{equation}
X=-\frac{1}{2}JM_dJ, \ \ J=I-\frac{1}{d\!+\!2}\mathbf{1}_{d\!+\!2}\mathbf{1}_{d\!+\!2}^T.  
\end{equation}
\State Since matrix $X$ is symmetric, it is diagonalizable, i.e., there exists a unitary matrix $S \!=\!(s_1,s_2,\cdots,s_{d\!+\!2}) \! \in \! \mathbb{R}^{(d\!+\!2) \times (d\!+\!2)}$ such that 
\begin{equation}\label{st}
    X = S \Lambda S^T,
\end{equation}
where  $\Lambda = \text{diag}(\lambda_1,\lambda_2,\cdots,\lambda_{d\!+\!2})$ is a diagonal matrix whose diagonal elements $\lambda_1 \ge \lambda_2 \ge \cdots \ge \lambda_{d\!+\!2}$ are eigenvalues. Let 
$S_*=(s_1, \cdots,s_d)$ and $\Lambda_*=\text{diag}(\lambda_1,\cdots,\lambda_d)$. Then, the congruent configuration $\check{q}=(q_i, q_{j_0}, \cdots, q_{j_d})^T$ 
can be obtained by
\begin{equation}\label{cong}
 \check{q} =  \Lambda_* ^{\frac{1}{2}}S_*^T.  
\end{equation}
\end{algorithmic}
\end{algorithm}

Inspired by the work \cite{han2017barycentric}, for the configuration $\check{p}=(p_i^T, p^T_{j_0}, \cdots, p^T_{j_d})^T$, 
we can obtain its congruent configuration $\check{q}=(q_i^T, q^T_{j_0}, \cdots, q^T_{j_d})^T$ by the following Algorithm \ref{disa} through distance matrix $M_d$. 
We will explain why $\check{q}$ in \eqref{cong} is the congruent configuration of $\check{p}$. Let \begin{equation}\label{md1}
\begin{array}{ll}
     &   z_h= p_h-\frac{1}{d+2}(p_i+\sum\limits_{k=0}^{d}p_{j_k}), \ h=i, j_0, \cdots, j_d, \\
     & Z=(z_i, z_{j_0}, \cdots, z_{j_d}), \\ 
     & b=(z_i^Tz_i, z^T_{j_0}z_{j_0}, \cdots, z^T_{j_d}z_{j_d})^T.
\end{array}
\end{equation}

Note that $z_i+\sum\limits_{k=0}^{d}z_{j_k}=0$. We have
\begin{equation}\label{md}
\begin{array}{ll}
     &     M_d = \mathbf{1}_{d\!+\!2}b^T+ b\mathbf{1}^T_{d\!+\!2}-2Z^T Z, \\
     & ZJ=Z(I-\frac{1}{d\!+\!2}\mathbf{1}_{d\!+\!2}\mathbf{1}_{d\!+\!2}^T)=Z, \\
     & X=-\frac{1}{2}JM_dJ = JZ^T ZJ= Z^T Z.
\end{array}
\end{equation}

From \eqref{md}, we can know that the matrix $X$ is positive semi-definite, i.e., the matrix  
$\Lambda_*$ in \eqref{cong} is also positive semi-definite. 
Note that  $\text{rank}(X) = \text{rank}(Z^T Z) \le \text{rank}(Z) \le d$. Based on \eqref{st} and \eqref{cong},  we have
\begin{equation}
    X = S \Lambda S^T = S_* \Lambda_* S_*^T= \check{q}^T\check{q}.
\end{equation}

Since $X=Z^T Z$ \eqref{md}, we have
$z_i^Tz_i=q_i^Tq_i$.  Then, we obtain 
\begin{equation}\label{md2}
    M_d =  \mathbf{1}_{d\!+\!2}b^T+ b\mathbf{1}^T_{d\!+\!2}-2Z^T Z = \mathbf{1}_{d\!+\!2}b^T+ b\mathbf{1}^T_{d\!+\!2}-2\check{q}^T\check{q}.
\end{equation}

From \eqref{md2}, we have
\begin{equation}
\| q_j \!-\! q_k\|_2= \| p_j \!-\! p_k\|_2, \ \text{for any} \ j,k \in \{i, j_0, \cdots, j_d \}.     
\end{equation}

Hence, the configurations $\check{q}$ and $\check{p}$ are congruent. From the above Lemma $4$, the displacement constraint \eqref{root1} is invariant to translations and rotations, i.e., congruent configurations $\check{p}$ and $\check{q}$ have the same displacement constraint. Hence, the displacement parameters $h_{ij_0}, h_{ij_1}, \cdots, h_{ij_d}$ in the displacement constraint $\sum\limits_{k=0}^{d}h_{ij_k} e_{ij_k} \!=\! \mathbf{0}$ can be calculated by solving the following matrix equation,
\begin{equation}\label{mi}
\left[ \!
\begin{array}{c c c c}
q_{j_0}\!-\!q_i & q_{j_1}\!-\!q_i  & \cdots &  q_{j_d}\!-\!q_i \\
\end{array}
\right]  \left[ \!
	\begin{array}{c}
	h_{ij_0} \\
	h_{ij_1} \\
	\vdots \\
	h_{ij_d}
	\end{array}
	\right] = \mathbf{0},
\end{equation}
where $q_i, q_{j_0}, \cdots, q_{j_d}$ is obtained by Algorithm \ref{disa}.

\subsection{Ratio-of-distance-based Displacement Constraint}

If the inter-agent ratio-of-distance measurements are available, we can obtain the ratio-of-distance matrix $M^r_d$ of follower $i$ and its $d\!+\!1$ neighbors ${j_0}, {j_1}, \cdots, {j_d}$ in $\mathbb{R}^d(d\ge2)$. 
For example, if follower $i$ is in 2-D space, i.e., $d=2$,
the ratio-of-distance matrix $M_2^r$ of follower $i$ and its three neighbors $j_0, j_1, j_2$ is given by
\begin{equation}\label{distance1}
M_2^r = \frac{1}{d_{j_0j_1}^2}
    \left[ \!
\begin{array}{c c c c }
0 & d_{ij_0}^2   & d_{ij_1}^2 &  d_{ij_2}^2 \\
d_{j_0i}^2 & 0  & d_{j_0j_1}^2 & d_{j_0j_2}^2 \\
d_{j_1i}^2 & d_{j_1j_0}^2 & 0 & d_{j_1j_2}^2  \\
d_{j_2i}^2 & d_{j_2j_0}^2 & d_{j_2j_1}^2 & 0
\end{array}
\right].
\end{equation}

Remark $5$. Under Assumption \ref{csr},
if the $d\!+\!1$ neighbors of follower $i$ get to their target positions, i.e., the $d\!+\!1$ neighbors ${j_0}, \cdots, {j_d}$  are not on a hyperplane in $\mathbb{R}^d$, we have $d_{j_0j_1} \neq 0$.

Based on the ratio-of-distance matrix $M_d^r$, we can obtain the similar configuration $\check{q}$ of the configuration $\check{p}$ 
by Algorithm \ref{disa}, where the distance matrix $M_d$ is replaced by the ratio-of-distance matrix $M_d^r$. From the above Lemma $4$, the displacement constraint \eqref{root1} is invariant to translations, rotations, and scalings i.e., similar configurations $\check{p}$ and $\check{q}$ also have the same displacement constraint. Hence, the displacement parameters $h_{ij_0}, h_{ij_1}, \cdots, h_{ij_d}$ in the displacement constraint $\sum\limits_{k=0}^{d}h_{ij_k} e_{ij_k} \!=\! \mathbf{0}$
can also be calculated by \eqref{mi}.

\begin{figure}[t]
\centering
\includegraphics[width=0.8\linewidth]{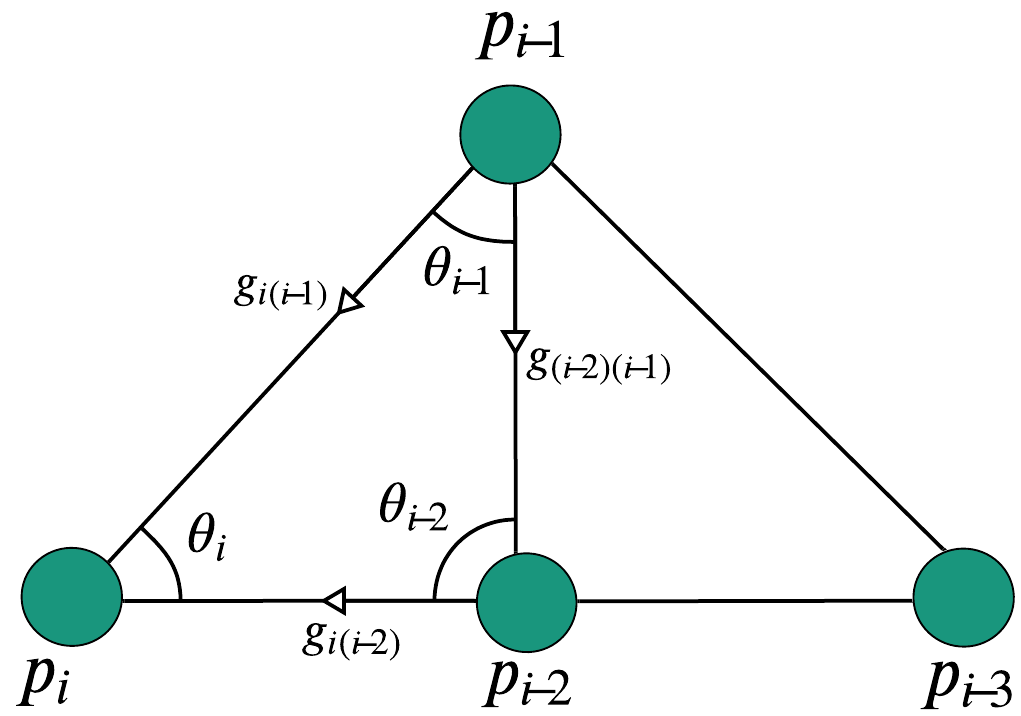}
\caption{2-D network.}
\label{2de}
\end{figure}

\subsection{Angle-based Displacement Constraint}

There are two cases:
\begin{enumerate}[(i)]
    \item  For follower $i$, 
    there is a neighboring agent $j \!=\! j_0, \cdots, \text{or}, j_d$ satisfying
    $\theta_{kih} \!=\! \theta_{kjh}, k, h \in \{j_0, \cdots, j_d \}, k \neq i  \neq h \neq j$;
    \item For follower $i$,
    there is no neighboring agent $ j \!=\! j_0, \cdots, \text{or}, j_d$ satisfying
    $\theta_{kih} \!=\! \theta_{kjh}, k, h \in \{j_0, \cdots, j_d \}, k \neq i  \neq h \neq j$.
\end{enumerate}

For the case (\romannumeral1), the displacement parameter $h_{ij}$ in \eqref{root1} is set as  $h_{ij}=1 (j \!=\! j_0, \cdots, \text{or}, j_d)$, and other displacement parameters in \eqref{root1} are set as $h_{ik}=0 (k \neq j, \ k = j_0, \cdots, j_d)$. For the case (\romannumeral2) under Assumption \ref{csr}, if the $d\!+\!1$ neighbors of follower $p_i$ get to their target positions, there must be an agent of $i, {j_0}, \cdots, {j_d}$ that can form a triangle with any other two agents. As shown in the above Fig. \ref{2de}, there is an agent ${j_0}$ that can form a triangle with any other two agents. For the triangle $\bigtriangleup p_ip_{j_0}p_{j_1}$, based on the sine rule, we can obtain the ratio-of-distances by using the angles shown as
\begin{equation}
 \frac{d_{ij_0}}{d_{j_0j_1}} = \frac{\sin \theta_{j_1}}{\sin \theta_{i}},  \ \ \  \frac{d_{ij_1}}{d_{ij_1}} =  \frac{\sin \theta_{j_1}}{\sin \theta_{j_0}},
\end{equation}
where $\theta_i, \theta_{j_1}, \theta_{j_1}$ are angles. Similarly, we can obtain the rest ratio-of-distances in $M_2^r$ \eqref{distance1} by the triangles  $\bigtriangleup p_ip_{j_0}p_{j_2}$ and $\bigtriangleup p_{j_0}p_{j_1}p_{j_2}$. Then, the displacement constraint of follower $i$ and its three neighbors $ {j_0}, {j_1}, {j_2}$ can be obtained by ratio-of-distance matrix $M_2^r$ through Algorithm \ref{disa} and \eqref{mi}. Hence, for the case (\romannumeral2), we can obtain the displacement constraint in $\mathbb{R}^d(d\ge2)$ by only using the angles.

\subsection{Bearing-based Displacement Constraint}\label{se2}

There are two cases:
\begin{enumerate}[(i)]
    \item  For follower $i$, 
    there is a neighboring agent $j \!=\! j_0, \cdots, \text{or}, j_d$ satisfying
    $g_{ik}=g_{jk}, k\in \{j_0, \cdots, j_d \}, k \neq i \neq j$;
    \item For follower $i$,
    there is no neighboring agent $j \!=\! j_0, \cdots, \text{or}, j_d$ satisfying
    $g_{ik}=g_{jk}, k\in \{j_0, \cdots, j_d \}, k \neq i   \neq j$.
\end{enumerate}

For the case (\romannumeral1), the displacement parameter $h_{ij}$ in \eqref{root1} is set as  $h_{ij}=1 (j \!=\! j_0, \cdots, \text{or}, j_d)$, and other displacement parameters in \eqref{root1} are set as $h_{ik}=0 (k \neq j, \ k = j_0, \cdots, j_d)$.
Under Assumption \ref{csr}, if the $d\!+\!1$ neighbors of follower $p_i$ get to their target positions, the positions of its $d\!+\!1$ neighbors $p_{j_0}, \cdots, p_{j_d}$  are not on a hyperplane in $\mathbb{R}^d$. Hence, 
for the case (\romannumeral2) under Assumption \ref{csr}, 
there must be an agent of $i, {j_0}, \cdots, {j_d}$ that can form a triangle with any other two agents. A 2-D example is given in the above Fig. \ref{2de}, where the three neighbors $ {j_0}, {j_1}, {j_2}$ of follower $i$ are not on a hyperplane (line) in $\mathbb{R}^2$. There is an agent ${j_0}$ that can form a triangle with any other two agents. For the triangle $\bigtriangleup p_ip_{j_0}p_{j_1}$, follower $i$ needs to obtain bearing $g_{j_0j_1}$ by communicating with agent $j_0$ or agent $j_1$. Then, follower $i$ can obtain the ratio-of-distances by using the bearings through the sine rule, i.e.,
\begin{equation}
 \frac{d_{ij_0}}{d_{j_0j_1}} = \frac{\sin \theta_{j_1}}{\sin \theta_{i}},  \ \ \  \frac{d_{ij_1}}{d_{ij_1}} =  \frac{\sin \theta_{j_1}}{\sin \theta_{j_0}},
\end{equation}
where
\begin{equation}\label{adf1}
\begin{array}{ll}
     &   \theta_{j_1} = \arccos (g_{ij_1}^Tg_{j_0j_1}), \  \theta_{i} = \arccos (g_{ij_1}^Tg_{ij_0}), \   \\
     & \theta_{j_0} = \arccos (g_{j_0i}^Tg_{j_0j_1}).
\end{array}
\end{equation}

Similarly, we can obtain the rest ratio-of-distances in $M_2^r$ \eqref{distance1} by the triangles  $\bigtriangleup p_ip_{j_0}p_{j_2}$ and $\bigtriangleup p_{j_0}p_{j_1}p_{j_2}$. Then, the displacement constraint of follower $i$ and its three neighbors $ {j_0}, {j_1}, {j_2}$ can be obtained by ratio-of-distance matrix $M_2^r$ through Algorithm \ref{disa} and \eqref{mi}. 
Hence, for the case (\romannumeral2), we can obtain the displacement constraint in $\mathbb{R}^d(d\ge2)$ by using the bearings. 

Note that the angles can also be obtained by local bearings, e.g., $\theta_{i} = \arccos ([g^{[i]}_{ij_1}]^Tg^{[i]}_{ij_0})$, where the superscripts $[i]$ represent the local coordinate frames of agent $i$.
That is,
the proposed method can be extended to local bearings. For example, If the bearing vector $g_{j_0j_1}$ can be obtained in the local coordinate frames of its neighbors $j_0,j_1 \in \mathcal{N}_i$, agent $i$ needs to communicate with its neighbor $j_0$ or $j_1$ to obtain the local bearing vector $g^{[j_0]}_{j_0j_1}$ or $g^{[j_1]}_{j_0j_1}$. Then, \eqref{adf1} becomes 
\begin{equation}
\begin{array}{ll}
     & \theta_{j_1} = \arccos ([g^{[j_1]}_{ij_1}]^Tg^{[j_1]}_{j_0j_1}), \  \theta_{i} = \arccos ([g^{[i]}_{ij_1}]^Tg^{[i]}_{ij_0}), \    \\
     &  \theta_{j_0} = \arccos ([g^{[j_0]}_{j_0i}]^Tg^{[j_0]}_{j_0j_1}).
\end{array}
\end{equation}

\subsection{Relative-position-based Displacement Constraint}

The proposed method can also be extended to local relative positions.
The local relative positions between 
follower $i$ and its $d\!+\!1 (d \! \ge \! 2)$ neighbors ${j_0}, \cdots, {j_{d}}$ are denoted by $e^{[i]}_{ij_0}(t),  \cdots, e^{[i]}_{ij_d}(t)$, respectively.
Note that $e_{ij_0}(t)=Q_i(t)e^{[i]}_{ij_0}(t),  \cdots, e_{ij_d}(t)=Q_i(t)e^{[i]}_{ij_d}(t)$, where $Q_i(t) \in SO(d)$ is the unknown rotation matrix of follower $i$. Then, the displacement constraint \eqref{root1} becomes
\begin{equation}\label{root2}
 \sum\limits_{k=0}^d Q_i(t)h_{ij_k}(t) e^{[i]}_{ij_k}(t) \!=\! \mathbf{0}.
\end{equation}

Note that $[Q_i(t)]^TQ_i(t)=I_d$. Although the rotation matrix $Q_i(t)$ is unknown, The displacement parameters $h_{ij_0}(t),  \cdots,  h_{ij_d}(t)$ in \eqref{root2} can be obtained by local relative positions through solving an equivalent equation of \eqref{root2} shown as
\begin{equation}\label{wmi}
\left[ \!
\begin{array}{c c c c}
e^{[i]}_{ij_0}(t) & e^{[i]}_{ij_1}(t) & \cdots & e^{[i]}_{ij_d}(t) \\
\end{array}
\right]  \left[ \!
	\begin{array}{c}
	h_{ij_0}(t) \\
	h_{ij_1}(t) \\
	\vdots \\
	h_{ij_d}(t)
	\end{array}
	\right] = \mathbf{0}.
\end{equation}

\subsection{Mixed-measurement-based Displacement Constraint}\label{sef}

If the multi-agent system is equipped with mixed types of measurements, i.e., some agents can measure only distances, while others can measure only local relative positions,  local bearings, angles, or
ratio-of-distances, the displacement constraint can also be obtained. The method of how to obtain the displacement constraint based on mixed measurements can be found in 
Section IV-B of \cite{fang20203}. Other measurements such as acoustic waves \cite{chen2020air} and measured absolute outputs \cite{zhou2022semiglobal} of the agents will be considered in the near future.

\subsection{Co-follower Group and Relaxed Graph}

\begin{figure}[t]
\centering
\includegraphics[width=0.8\linewidth]{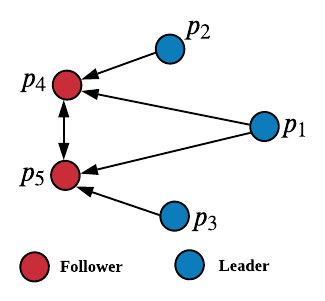}
\caption{$p_4, p_5$ forms a co-follower group.}
\label{2dneighbor}
\end{figure}

In this article, only the first follower has $d\!+\!1$ leaders as neighbors in $\mathbb{R}^d$.
The proposed method can be extended to the case that no follower has $d\!+\!1$ leaders as neighbors yet still localizable. One example in $\mathbb{R}^2$ is given in Fig. \ref{2dneighbor}, where
the leaders $1, 2, 3$ are noncolinear.
The followers $4, 5$ can measure local relative positions. Follower $4$ can only obtain the local relative positions from two leaders $1, 2$, while follower $5$ can only obtain the local relative positions from two leaders $1, 3$, i.e., no follower has three leaders as neighbors. But the followers $4$ and $5$ can still be localized by the leaders $1, 2, 3$ through exchanged information, i.e., followers $4$ and $5$ exchange their available information. Then, the followers $4, 5$ can implement the proposed controllers for the followers to achieve distributed localization and formation maneuver control.

\begin{defn}\label{dfkl1}
A group of followers $\mathcal{V}_{c} = \{{i_1}, \cdots, {i_z}\}$  is
called a co-follower group if 
the following conditions hold: 

(1) The cardinality of the neighbors of a co-follower group is at least $d\!+\!1$ in $\mathbb{R}^d$. For example, there are three neighbors $1, 2, 3$ of the co-follower group $4, 5$ in $\mathbb{R}^2$ shown in Fig. \ref{2dneighbor};

(2) The followers in $\mathcal{V}_c$ can be  localized by their neighbors through exchanged information if their $d\!+\!1$
neighbors are not on a hyperplane in $\mathbb{R}^d$.    
\end{defn}

\begin{defn}\label{dfkl2}
If follower $i$ does not belong to any co-follower group, follower $i$ is called single-follower.    
\end{defn}

Then, the graph condition can be relaxed based on the concepts of co-follower group and single-follower given below:
\begin{enumerate}[(1)]
\item The agents in $\mathcal{G}$ are divided into $\kappa \!>\! 1$ subsets $\mathcal{V}_1,  \mathcal{V}_{2},$ $\cdots, \mathcal{V}_{{\kappa}}$, where $\mathcal{V}_i \cap \mathcal{V}_j = \emptyset$ if $i \neq j$. Agent $i$ is called in layer $g$ if  $i \in \mathcal{V}_{g} (1 \le g \le \kappa)$. Subset $\mathcal{V}_1$ includes all leaders, i.e., $\mathcal{V}_1=\mathcal{V}_l$. The union of the subsets $\mathcal{V}_{2}, \cdots, \mathcal{V}_{{\kappa}}$ include all followers, i.e.,  $\bigcup\limits_{g=2}^{\kappa} \mathcal{V}_g = \mathcal{V}_f$;

\item Each single-follower or co-follower group  only belongs to one follower subset $\mathcal{V}_{g}( g=2, \cdots, \kappa$);
\item Each single-follower $i$ 
is $d\!+\!1$-reachable from the leader set $V_1$, and its neighbor set ${N}_i$ is given by
\begin{equation}
\mathcal{N}_i = \{ j \in \bigcup\limits_{s=1}^{g} \mathcal{V}_{s}:  j <i, \ (i,j) \in \mathcal{E}, \ i \in \mathcal{V}_{g} \}.
\end{equation}

\item Each co-follower group 
is $d\!+\!1$-reachable from the leader set $V_1$, and the neighbor set $\mathcal{N}_c$ of each co-follower group $\mathcal{V}_{c} = \{{i_1}, \cdots, {i_z}\}$  is given by 
\begin{equation}\label{neig2}
\begin{array}{ll}
     & \mathcal{N}_c = \{ j \in \bigcup\limits_{s=1}^{g}:  (i_s,j) \in \mathcal{E}, \ j<i_s, \\
     & \ \ \ \ \ \ \ \ \  i_s \in \mathcal{V}_g, \ s=1, \cdots, \text{or} \ z \}.   
\end{array}
\end{equation}

\end{enumerate}

\subsection{Continuous Velocity Controller}

To guarantee the continuity of velocity controllers of the followers, we only need to guarantee the continuity of velocity controllers of the followers when they switch between maintaining mode and maneuvering mode.

\begin{lemma}
    The velocity controller of the follower $i \in \mathcal{V}_f$ in $\mathbb{R}^d$ is continuous when it switches between maintaining mode and maneuvering mode if 
the parameters $a_1$ in (45) and $a_2$ in (50)  are designed as
\begin{equation}\label{switch}
 a_1 = \text{diag}(\xi), \ \ a_2 = \text{diag}(\varsigma), 
\end{equation}
where 
\begin{equation}
\begin{array}{ll}
    & \xi = [\xi_{1}, \cdots, \xi_d],  \ \ \varsigma =[\varsigma_{1}, \cdots, \varsigma_d], \\
     & [\eta_{i1}, \cdots, \eta_{id}]^T =  a_2 \hat p_i \!+\!  a_2\sum\limits_{j \in \mathcal{N}_i} \frac{w_{ij}}{w_{ii}}{\hat p}_{j}, \\
        &   \frac{\xi_{k}}{\varsigma_{k}}\eta_{ik} = \eta_{ik} + \text{sig}^{a_3}(\eta_{ik} ), \ \xi_{k}, \varsigma_{k} >0,  \ k= 1, \cdots, d.
\end{array}
\end{equation}
\end{lemma}

\subsection{Further Comparison with the Existing Results}\label{sefw}

The further comparison with the existing non-relative-position-based results \cite{mehdifar20222, cao2011formation, jiang2016simultaneous,nguyen2019persistently,han2018integrated, cao2019relative,yang2020distributed,chen2022simultaneous}
are given below: 

\begin{enumerate}[(i)]

    \item From the "global shape convergence" point of view,  the local convergence is realized in \cite{cao2011formation, jiang2016simultaneous}, and the (almost) global convergence is realized in this work and \cite{mehdifar20222, cao2019relative,han2018integrated}. The work in \cite{nguyen2019persistently} achieves semi-global convergence. The asymptotic convergence in \cite{yang2020distributed,chen2022simultaneous} is neither local nor global.

    \item  From the "assumptions on the formation graph" point of view, the formation graphs are undirected in \cite{jiang2016simultaneous, chen2022simultaneous} and directed in \cite{mehdifar20222, cao2011formation,nguyen2019persistently,han2018integrated, cao2019relative,yang2020distributed}  and this work.

    \item From the "coordinate frame" point of view, the work in \cite{mehdifar20222} is applicable to unaligned local coordinate frames, and the works in \cite{cao2011formation,nguyen2019persistently,han2018integrated, jiang2016simultaneous, cao2019relative,yang2020distributed,chen2022simultaneous} and this work are applicable to global coordinate frame.

\end{enumerate}

\ifCLASSOPTIONcaptionsoff
  \newpage
\fi





\bibliographystyle{IEEEtran}
\bibliography{papers}


\end{document}